\documentclass[lettersize,journal]{IEEEtran}
\usepackage{amsmath,amsfonts}
\usepackage{algorithmic}
\usepackage{algorithm}
\usepackage{array}
\usepackage[caption=true,font=normalsize,labelfont=sf,textfont=sf]{subfig} % caption was set to false initially
\usepackage{textcomp}
\usepackage{stfloats}
\usepackage{url}
\usepackage{verbatim}
\usepackage{graphicx}
\usepackage{cite}

\usepackage{amsthm} % for proofs

\newtheorem{theorem}{Theorem}
\newtheorem{definition}{Definition}
\newtheorem{corollary}{Corollary}
\newtheorem{example}{Example}
\newtheorem{lemma}{Lemma}

\newcommand\R{\mathbb{R}} % Real numbers
\newcommand\Rplus{\mathbb{R}^{+}} % Positive real numbers
\newcommand\F{\mathcal{F}} % Function sets
\newcommand\T{\mathcal{T}} % Tandems of servers
\newcommand\server{\mathcal{S}} % Server: s, \mathcal{S}
\newcommand\flow{f} % Flow: f, f
\newcommand{\func}{g} % function: f, g, h
\newcommand{\funcother}{h} % function: f, g, h
\newcommand{\hdev}{hDev} % For horizontal deviation: h, hDev, h_dev
\newcommand{\vdev}{vDev} % For vertical deviation: v, vDev, v_dev
\newcommand{\1}[1]{\mathbf{1}_{\{#1\}}} % For the indicator or test function: \mathbf{1}, \mathbb{1}
\newcommand{\curveclass}{MSLC} % Name of the curve class.

\begin{document}

\title{LUDB++: Enabling LUDB for the Analysis of Shaped Feedforward FIFO Networks using Network Calculus}

\author{Alexander Scheffler, \\Faculty of Mathematics and Computer Science, \\University of Hagen, Germany{} 
%\thanks{Manuscript received April 19, 2021; revised August 16, 2021.}}
}



\maketitle

\begin{abstract}
% 150 - 250 words: 184
This paper discusses how latency guarantees for non-cyclic (feedforward) First-In-First-Out (FIFO) networks with shapers can be computed within the Network Calculus (NC) framework.
Shapers are methods implemented in software or hardware and may reside inside the network and at the endpoint which constrain the rate and maximum packet sizes for the transmission of specific data streams (flows) or groups thereof.
Shaping can improve latencies and is an important aspect of Time-Sensitive Networking (TSN).
Several methods in NC exist to analyze FIFO networks.
Among them is the Least Upper Delay Bound (LUDB) methodology.
So far, LUDB does not incorporate shaping assumptions into its analysis.
This paper addresses this gap resulting in the new methodology called LUDB++.
The evaluation on a set of different line topologies and a tree topology with a total of 130 configurations shows that LUDB++ delivers more accurate latency bounds compared to LUDB.
Moreover, the Exponential Linear Program (ELP) method, which considers FIFO and shaping inside the network, yields the most accurate bounds to this date.
ELP is superseded by LUDB++ for most of cases by a margin of up to 9.13$\%$.
\end{abstract}


\section{Introduction}
\IEEEPARstart{R}{eal-time} systems inclusive of the automotive industry, avionics and industrial automation require deterministic guarantees on timing behavior.
Missing a response within a certain timeframe can often be critical in such systems, e.g., a delayed triggering of a brake after pressing the respective pedal.
Ethernet as communication backbone for such systems has become increasingly popular in the last years.
It is low-cost and an enhancement of the standard called TSN \cite{ieee802.1Q} offers reliable network communication with deterministic delays.
TSN also comes with traffic shaping which limits the rate and maximum packet sizes certain data flows can transmit in the network.
This method can reduce maximum delays inside the network, especially for high-priority flows.
NC \cite{cruz2002calculus1} \cite{cruz2002calculus2} has been proven useful for the formal derivation of delays. 
This paper assumes that FIFO scheduling -- which is widespread and easy to implement -- alongside shaping inside the network and at the endpoint is at place.
Limited work in NC exist that studies both aspects in conjunction. 
One approach is ELP \cite{bouillard2022trade} which constructs one global Linear Program (LP) as upper bound of the Mixed Integer Linear Program (MILP) discussed in \cite{bouillard2014exact} by removing the integer variables of the latter.
Different to the MILP of \cite{bouillard2014exact}, ELP also includes line shaping constraints and thus can be seen as the most accurate analysis to this date for our studied network model.
Another stream of research in NC for FIFO is LUDB \cite{lenzini2008methodology} \cite{scheffler21qest}.
It is a modular analysis where tandems (servers in a row) are analyzed separately and where the global solution consists of stitching the individual solutions together.
Its complexity mainly stems from (\romannumeral 1) dividing up the network into tandems and (\romannumeral 2) solving each tandem that occurred during the analysis by means of solving a set of LPs.
So far, LUDB is restricted to simple, so-called token bucket, arrival curves and pseudo-affine service curves \cite{lenzini2008methodology} where a rate-latency curve is a special form of it.
LUDB does not take shaping into account -- neither at the endpoint nor inside the network.
In this paper, we propose LUDB++, an extension of LUDB that considers shaping assumptions.
For this, step (\romannumeral 2) had to be completely revised as the curves during the analysis got significantly more complex.
Our evaluation on different tandems and a tree topology reveals that LUDB++ delivers more accurate bounds than ELP for almost all cases.
As the modeling complexity increased compared to LUDB, so did its runtime. 
Nonetheless, we were able to carry out our evaluation with topologies where the tandems during the analysis were each spanning up to 8 servers. 
Additionally, we provide an open-source C++ implementation of LUDB++.  \\ 
The rest of this paper is organized as follows: 
Section \ref{sec:nc_background} gives background on NC and states the system model for this study.
In Section \ref{sec:related_work}, we report the related work. 
Section \ref{sec:ludb_plus_plus} presents the LUDB++ methodology including a comprehensive example after which it is evaluated on different network configurations in Section \ref{sec:eval}.
The paper ends with concluding remarks in Section \ref{sec:conclusion}.

\section{Network Calculus Background and System Model}\label{sec:nc_background}
An overview on NC performance analysis can be found in \cite{bouillard2018deterministic} \cite{le2001network}.
\\
\begin{definition}[Cumulative Data Functions]\label{def:fun}
Let function $A$ describe the cumulative input data arrival of a flow crossing a server $\server$.
Then $A^{\prime}$ describes its cumulative output after $\server$.
These functions are in the set $\F_0 := \{ f : \R_\infty \rightarrow \Rplus_\infty  | f(0) = 0 $ $\forall s \leq t: f(t) \geq f(s) \}$ with $\R_\infty = \R\cup\{\infty\}$ and $\Rplus_\infty = \Rplus\cup\{\infty\}$, respectively. 
\end{definition}

\begin{definition}[Arrival Curve]\label{def:ac}
Let $A$ be the cumulative input function of a flow $\flow$ at server $\server$.
Then we call $\alpha \in \F_0$ an arrival curve of $f$ at $\server$ if
\[
\mbox{$\forall \, 0 \leq d \leq t: A(t) - A(t-d) \leq \alpha(d)$.}
\]
\end{definition}

\begin{definition}[Service Curve]\label{def:sc}
Let $A$ be an input to a server $\server$ and $A^{\prime}$ be the respective output function.
Then we say that $\beta \in \F_0$ is a service curve for $\server$ if
\[
\mbox{$\forall t \geq 0 : A^{\prime} (t) \geq \inf\limits_{0 \leq s \leq t} \{A(t-s) + \beta(s) \} =: A \otimes \beta (t)$.}
\]
\end{definition}

\begin{definition}[NC Operations]
\label{def:nc:min_plus_operations}
The (min,plus)-algebraic aggregation, convolution and deconvolution
of the functions $\func,\funcother\in\F_0$ are defined as
\begin{eqnarray}
\text{aggregation:} &  & \hspace{-6.5mm}\left(\func+\funcother\right)\left(d\right) = \func\left(d\right)+\funcother\left(d\right)\!,\\
\text{convolution:} &  & \hspace{-6.5mm}\left(\func\otimes \funcother\right)(d) = \hspace{-1.5mm}{\displaystyle \inf_{0\leq u\leq d}}\hspace{-0.75mm}\left\{ \func(d-u)+\funcother(u)\right\} \hspace{-0.75mm},\\
 \hspace{-5.5mm}\text{deconvolution:} &  & \hspace{-6.5mm}\left(\func \oslash \funcother\right)(d) = \sup_{u\geq0}\left\{ \func(d+u)-\funcother(u)\right\}\!.
\end{eqnarray}
\end{definition}

\begin{theorem}[Performance Bounds]
\label{thm:Basic_Bounds}\label{def:deconvolution}
Consider a server $\server$ that offers a service curve $\beta$.  
Assume flow $\flow$ has arrival curve $\alpha$. 
Then we obtain the following bounds:
\begin{eqnarray*}
	\text{Output Bound: } \alpha^{\prime} (t) = \alpha \oslash \beta (t) := \sup\limits_{ u \geq 0 } \{ \alpha(t + u) - \beta(u) \} \\
		\text{Delay Bound: } \hdev(\alpha, \beta) = \inf\{d \geq 0: (\alpha \oslash \beta) (-d) \leq 0  \} \\
		\text{Backlog Bound: } \vdev(\alpha, \beta) = \sup\limits_{ u \geq 0 } \{ \alpha(u) - \beta(u) \} = \alpha \oslash \beta (0)
\end{eqnarray*}
where $\alpha^{\prime}$ is a bound on $A^{\prime}$, the output from server $\server$ (see Definition \ref{def:sc}), and the horizontal (vertical) deviation $\hdev$ ($\vdev$) between arrival curve and service curve bounds the delay (backlog) experienced by $f$ at $\server$.
\end{theorem}

\begin{theorem}[Convolution of Service Curves]
\label{thm:convolution on a tandem}
Consider two servers in tandem $\server_1$ and $\server_2$ with service curves $\beta_1$ and $\beta_2$ respectively for flow $\flow$. Then, $\beta_1 \otimes \beta_2$ is a service curve for $f$ for the system $(\server_1, \server_2)$.
\end{theorem}

\begin{theorem}[FIFO Leftover Service]\label{thm:fifo_lo}
Let server $\server$ offer service curve $\beta$.  
Assume flows $\flow_1$ and $\flow_2$ with arrival curves $\alpha_1$ and $\alpha_2$ cross $\server$.
Assuming FIFO multiplexing, the leftover service for $f_1$ is
\[
\mbox{$\beta_{\flow_1}^{\text{l.o.}}(t) = [\beta(t) - \alpha_2(t-\theta)]^+ \cdot \1{t > \theta}$}
\]
with 
$[x]^+ := \max \{0,x\}$,
the indicator function $\1{  \text{condition} }$  that is $0$ if the condition is not met and $1$ otherwise,
and $\theta \in \Rplus$ is the free FIFO parameter.
As abbreviation for the expression we use $\beta \ominus_{\theta} \alpha_2$.
\end{theorem}

\begin{example}[Common Curve Shapes]
Common curve shapes in NC include the so-called burst-delay function $\delta_T(t)$ that is $0$ for $t\leq T$ and $\infty$ otherwise.
$\delta_0(t)$ is the neutral element w.r.t. operator $\otimes$. 	
A typical service curve is the so-called rate latency curve defined as $\beta_{R,T}(t)=R \cdot [t-T]^+$.
A common arrival curve is the so-called token bucket curve defined as $\gamma_{b, r}(t)=r \cdot t + b$.
\end{example}

\begin{definition}[Shapers]
The cumulative output $A^{\prime}$ after server $\server$ can further be constrained by a so-called shaper.
It can be described by a curve $\gamma_{L,R^{\prime}}$ where $R^{\prime}$ typically stands for the maximum transmission capacity of the outgoing link and $L$ for the maximum packet size.
A shaper may also be assumed to hold at the entry of the network, i.e., the endpoint itself shapes some of its outgoing flows.
Mathematically, the output process is  constrained by $(\alpha \oslash \beta) \wedge \gamma_{L,R^{\prime}}$ %(\cite{bouillard2018deterministic}, Theorem 5.3)
after server $\server$ and by $\alpha \wedge \gamma_{L,R^{\prime}}$ after the endpoint with $a \wedge b := min\{a,b\}$.
\end{definition}

Our network model consists of data flows originating from endpoints.
Each flow is constrained by a token bucket arrival curve with token bucket endpoint-shaping. 
The flows traverse a network of FIFO-multiplexing servers.
Each server offers a rate latency curve\footnote{more general curve is allowed as long as it fits into the curve class described in Section \ref{sec:ludb_plus_plus}. }  and a token bucket output shaper per outgoing link.
For stability, we assume per traversed server $\server$ that the sum of the of the respective flows' non-shaped arrival curve rates is less than the rate of $\server$.
Moreover, we assume that each shaper rate $R^{\prime}$ is at least as high as the service rate of the corresponding server $\server$.
Additionally, we assume that $R^{\prime}$ of $\server$ is also at least as high as the following service rates of servers on any path of a flow that traverses $\server$ on the respective link.
The assumption also applies to the endpoint-shaping rate.
A similar assumption on the shaper rate can be found in \cite{thomas2019cyclic}.
Lastly, we limit our study to feedforward networks, so the paths among the flows do not exhibit cyclic dependencies.

\section{Related Work}\label{sec:related_work}
\subsection{TFA, TFA++, FP-TFA, GFP-TFA, AdmTFA} 
Total Flow Analysis (TFA) \cite{bouillard2018deterministic} is a method to compute a delay bound for the flow of interest (foi) by iteratively applying delay and output bound of Theorem \ref{thm:Basic_Bounds} to the aggregation of flows (Definition \ref{def:nc:min_plus_operations}) along foi's path.
TFA has low computational cost, however, the bounds become increasingly loose with increasing network size.
TFA++ \cite{grieu2004} \cite{mifdaoui2017beyond} is an extension of TFA that incorporates shaping constraints of the links.
For this, the respective shaper curve is used to get a tighter input arrival curve when computing the node's delay bound, i.e., $\hdev(\alpha \wedge \gamma_{L,R^{\prime}}, \beta)$ instead of $\hdev(\alpha, \beta)$.
Similarly, the shaper curve is used go get a tighter output bound.
Other extensions of TFA include Fixed-Point TFA (FP-TFA) \cite{thomas2019cyclic} for cyclic networks, Generic FP-TFA (GFP-TFA) \cite{tabatabaee2023efficient} for more general curves than TFA-FP and Admission Shaping with TFA (AdmTFA) \cite{bouillard2024admission}.
AdmTFA solves TFA as an LP that incorporates shaping constraints at the entry of the network and within the network. 
It is also applicable to cyclic networks for which a fixed point equation is computed.

\subsection{SFA-FIFO} 
Separated Flow Analysis with FIFO (SFA-FIFO) \cite{scheffler21qest}  computes the residual service curve at each server along the path of the foi according to Theorem \ref{thm:fifo_lo}. 
The $\theta$ value is fixed as $T+\frac{b}{R}$ for service curve $\beta_{R,T}$ and arrival curve $\gamma_{b,r}$.
Afterwards, Theorem \ref{thm:convolution on a tandem} is applied iteratively on the residual service curves.
From this, the final delay bound is computed with Theorem \ref{thm:Basic_Bounds}. 
It does not take shaping into account, has low computational cost and often results in better bounds than TFA.

\subsection{LUDB, LUDB-FF}
LUDB \cite{lenzini2008methodology} was initially developed for nested tandems which are defined as follows.
\begin{definition}[Nested Tandem]\label{def:nested_tandem}
Tandem $\T$	consisting of a set of subsequent servers $\server_1, \server_2, ..., \server_n$ is called a nested tandem iff for all pairs of flows $\flow_i, \flow_j$ its respective paths on $\T$  either have a subset relation or are fully disjoint.
\end{definition}

Note that the previous definition assumes that a flow does not rejoin $\T$.
If the flow indeed rejoins $\T$, the flow is split into several flows for each rejoining segment on $\T$.

A nested tandem allows the analysis to proceed from the inside out applying Theorem \ref{thm:fifo_lo} and Theorem \ref{thm:convolution on a tandem} along the way.
For example, for the following nested tandem $\T$,

\begin{figure}[H]  
\begin{centering}
  \includegraphics[width=0.5\textwidth]{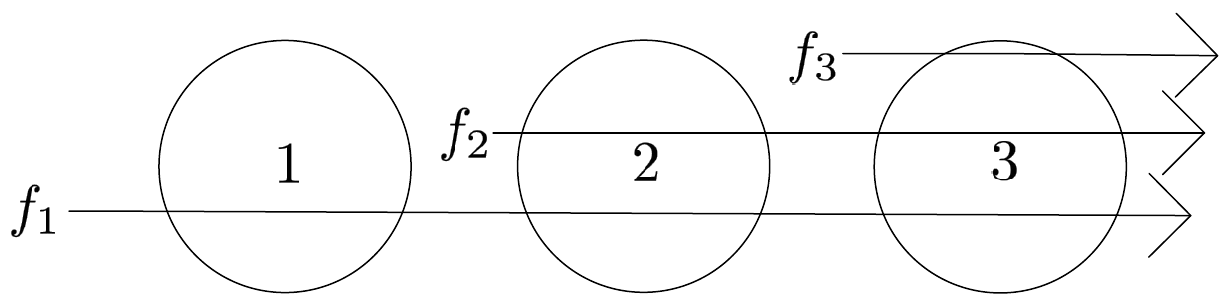}  
  \caption{Sinktree tandem with $N=3$.}
  \label{fig:eval_3_servers_sink_tree_tandem}
\end{centering}
\end{figure}
the analysis would proceed as follows. 
$\beta_3 \ominus_{\theta_1} \alpha_3$ is what is left at  $\server_3$ for the aggregate $\flow_2$ and $\flow_3$.
This result is then convolved with the service curve of $\server_2$, i.e., $\beta_2 \otimes (\beta_3 \ominus_{\theta_1} \alpha_3)$ which is how much service is left for the aggregate $\flow_2$ and $\flow_3$ on the subtandem $\T_{2,3}$ that spans $\server_2$ and $\server_3$ only.
That result in turn is used to remove flow $\flow_2$'s interference on $\T_{2,3}$ given by the result of $(\beta_2 \otimes (\beta_3 \ominus_{\theta_1} \alpha_3)) \ominus_{\theta_2} \alpha_2$.
This can be convolved with the service curve of $\server_1$ which results in $ \beta_1 \otimes ((\beta_2 \otimes (\beta_3 \ominus_{\theta_1} \alpha_3)) \ominus_{\theta_2} \alpha_2)$ that is the leftover service curve for $\flow_1$ for the complete tandem $\T$ denoted by $\beta_{\flow_1}^{\text{l.o.}}$.
Finally, the delay bound can be computed with Theorem \ref{thm:Basic_Bounds}, namely $hDev(\alpha, \beta_{\flow_1}^{\text{l.o.}})$.

LUDB automates this procedure by first building up a so-called nesting tree that encodes the steps of the analysis in a tree which is then traversed bottom-up.
Moreover, LUDB does not set the FIFO parameters $\theta_i$ to a fixed value during the traversal of the nesting tree.
Instead, it formulates all steps of the analysis symbolically to then be able to optimize the $\theta_i$'s to achieve the lowest upper delay bound w.r.t. the involved NC operation in this order.
LUDB is tight\footnote{i.e., there exist a sample path that coincides with the upper bound} for sink tree networks \cite{lenzini2006tight}, but the computed bound can be larger than the actual worst case delay for other topologies.
Another limitation of LUDB is that it is restricted to token bucket arrival curves and pseudoaffine service curves for which the rate latency curve is a special form of it.
Later on, LUDB has been extended to non-nested tandems \cite{bisti2012numerical} by virtual flow splitting to transform non-nested tandems into nested ones.
LUDB has also been extended to feedforward networks denoted by LUDB-FF \cite{scheffler21qest}.

\subsection{LUDB with Shaped Foi}
\cite{boyer2010half} presents an extension of LUDB where the foi experiences a so-called concave piecewise linear (CPL) arrival curve.
The study is limited to one hop persistent cross-traffic at each server, so a special nested topology.
The crossflows' arrival curves are all limited to token bucket arrival curves.
The approach is quite similar to LUDB with the difference that it finds the optimal $\theta_i$ values w.r.t. delay bound of a CPL arrival curve (new) and left-over service curve for this topology (old).
It derives a close form solution for the delay bound which is then at least as good as LUDB.
This is because LUDB ignores shaping and more generally, reduces CPL arrival curves to token bucket curves.
The authors of \cite{boyer2010half} admit that "the general case, with shaping on considered and interfering flow ... seems really harder", exactly where this paper contributes to.

\subsection{ELP, PLP and PLPP} 
As LUDB does not guarantee tight delay bounds, \cite{bouillard2014exact} presents an approach based on mathematical optimization that provably delivers tight delay bounds in FIFO feedforward networks without shapers. 
For this purpose, it encodes the complete network with its flow paths, arrival curves constraints, service curve constraints, FIFO property, timing relations, monotonicity constraints and indicator variables into a single exponentially-sized MILP.
The MILP has not been researched with shaping constraints to this date.
Recent works discussed in \cite{bouillard2022trade} relax the tightness guarantee as follows.
ELP is the MILP of \cite{bouillard2014exact} without the indicator variables, thus effectively turning the MILP into a LP of exponential size.
Polynomial Linear Program (PLP) is a polynomial-sized LP that takes the main ideas of the MILP. 
It has fewer time variables than the MILP which are naturally ordered and thus do not require indicator variables that enforce such an ordering.
Moreover, PLP includes TFA++ and SFA constraints to make the delay bound result less pessimistic.
Additionally, PLP incorporates shaping constraints within the network, i.e., line shaping constraints\footnote{Note that our evaluation in Section \ref{sec:eval} shows that the implementation of ELP which is available at \url{https://github.com/anne-bou/panco} seems to also incorporate shaping assumptions, although this is not officially stated in \cite{bouillard2022trade}. }.
PLPP (PLP' in \cite{bouillard2022trade}) is the same as PLP just without the TFA++ and SFA constraints and thus an upper bound of PLP's delay results.

\section{LUDB++}\label{sec:ludb_plus_plus} 
% Structure:
% Def. window function I_w => DONE
% Def. of \curveclass for n=1: term, figure, acknowledgment => DONE
%  We call the curve class that can be described as minimum of steps each followed by a linear curve \curveclass. => DONE
% It might be tempting to use convolution from n=1 to larger number to model multiple steps, but note that it cant just be replaced by MIN here unlike the curve shapes within LUDB (pseudoaffine called there). => NOT NEEDED
% Def. of \curveclass for arbitraty n => DONE
% Figure of it? => DONE
% RL in \curveclass => DONE
% \curveclass closed under \conv => DONE
% delay with shaped ac => DONE
% backlog with non-shaped ac => DONE
% output bound => DONE
% LO with shaped ac in MSLC => DONE
% symbolic LO with shaped AC description => DONE
% with theorems ... we are now able to analyze ....and give an example in the following. => DONE
% example shaped ac and delay for two servers -> use some abbreviations to make it compact => DONE
This Section constitutes our main contribution, LUDB++, an extension of LUDB that incorporates shaping assumptions into its analysis.
While the order of operations of the analysis is the same as in LUDB (discussed in Section \ref{sec:related_work}), the statements of intermediary results need to be completely revised.
We first start by defining a useful function, the so-called window function $I_w$ \cite{chang2012performance}.
\begin{definition}[Window Function]\label{def:window_function}
The window function is defined as 
$I_w (t) =     \begin{cases}
        w & \text{if } t=0\\
        +\infty & \text{if } t>0
    \end{cases}$
\end{definition}
Next, we model an essential curve for our study within the NC framework that begins at some offset, followed by a step after which it continues linearly.
This can be modeled with the curve 
\begin{equation} \label{eq:plcs_n_1}
	((\delta_{\tau} \otimes I_{\sigma_1} \otimes \gamma_{0,\rho_1})\wedge \delta_0)\otimes \delta_D
\end{equation} since

$\inf_{\substack{s,u\geq 0 \\ s+u=t}} \{ (\delta_{\tau}\otimes I_{\sigma_1} \otimes \gamma_{0, \rho_1})\wedge \delta_0 (s) + \delta_D (u) \} $\\
$= \begin{cases}
        0 & \text{if } t\leq D  \\
        ((\delta_{\tau} \otimes I_{\sigma_1} \otimes \gamma_{0,\rho_1})\wedge \delta_0)(t-D) & \text{if } t>D 
        \end{cases}
        $
 
Note, that $(\delta_{\tau} \otimes I_{\sigma_1} \otimes \gamma_{0,\rho_1})\wedge \delta_0$ is a function that starts at the origin, has a plateau of height $\sigma_1$ with width $\tau$ and is then followed by a linear rate of $\rho_1$.

We call the curve class that can be described as minimum of such functions that all begin after a common offset $\curveclass$ as abbreviation for 'Minimum of Steps Linear Curve'. %(Minimum of Steps Linear Curve)
Figure \ref{fig:mslc_n_1} illustrates this class for $n=1$.
\begin{figure}[H]  
\begin{centering}
  \includegraphics[width=0.5\textwidth]{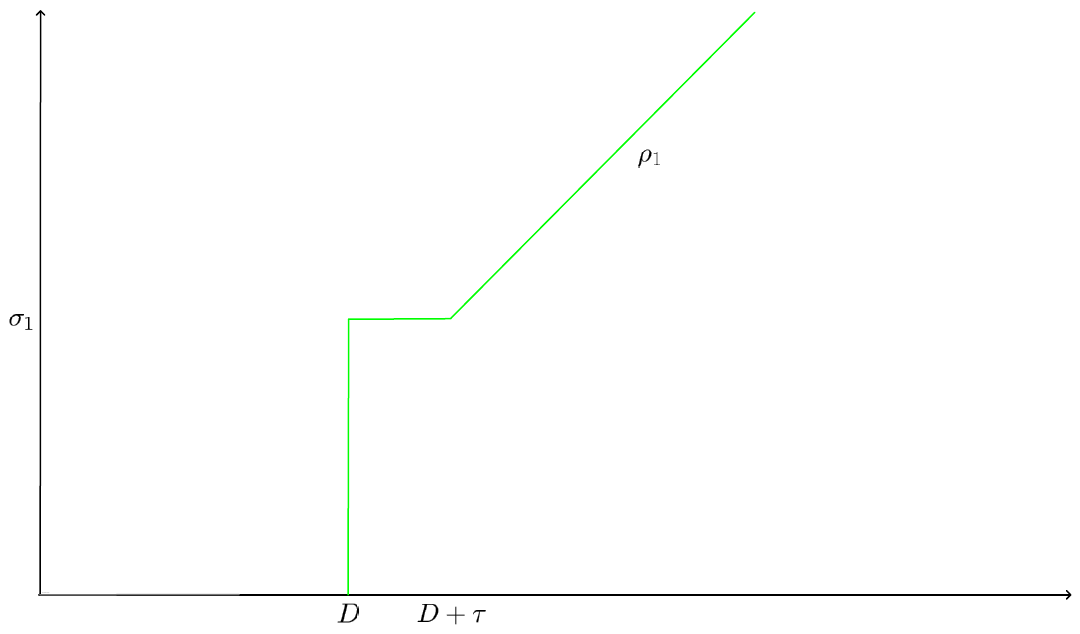}  
  \caption{$\curveclass$ with $n=1$.}
  \label{fig:mslc_n_1}
\end{centering}
\end{figure} 

\begin{definition}[\curveclass]\label{def:mslc}
The service curve $\beta \in$ \curveclass~is defined as $\beta = \delta_D \otimes \bigwedge\limits_{i=1}^n ((\delta_{\tau_i} \otimes I_{\sigma_i} \otimes \gamma_{0,\rho_i})\wedge \delta_0)$.
\end{definition}

\begin{lemma}[Rate Latency Service Curve is in \curveclass]\label{lemma:rl_in_mslc}
Let $\beta=\beta_{R,T}$ be a rate latency service curve. 
Then, $\beta \in$ \curveclass.
\end{lemma}
 
\begin{theorem}[\curveclass~Closed under Convolution]\label{theorem:mslc_closed_under_conv}
	Let $\beta_1, \beta_2 \in$ \curveclass, i.e., $\beta_1=\delta_{D_1}\otimes (\bigwedge\limits_{i=1}^{n_1}(\delta_{\tau_i} \otimes I_{\sigma_i} \otimes \gamma_{0, \rho_i})\wedge \delta_0)$ and $\beta_2=\delta_{D_2}\otimes (\bigwedge\limits_{i=1}^{n_2}(\delta_{\tilde{\tau_i}} \otimes I_{\tilde{\sigma_i}} \otimes \gamma_{0, \tilde{\rho_i}})\wedge \delta_0)$. 
	Then $\beta_1 \otimes \beta_2 \in$ \curveclass~as well.
\end{theorem}

The delay bound for shaped arrival curves and this new curve class can be computed as follows.

\begin{theorem}[Delay Bound Shaped Arrival Curve and \curveclass~Service Curve]\label{theorem:delay_bound_shaped_ac_plcs_sc}
	Let $\alpha=\gamma_{b,r} \wedge \gamma_{L,R'}$ and $\beta \in$ \curveclass~of the form $\beta=\delta_D \otimes (\bigwedge\limits_{i=1}^n (\delta_{\tau_i}\otimes I_{\sigma_i} \otimes \gamma_{0,\rho_i}) \wedge \delta_0)$ with $L \leq b$ and $r < R', R' \geq \max\limits_{i=1,...,n} \{ \rho_i | \rho_i < +\infty\}, r \leq \min\limits_{i=1,...,n} \rho_i$. 
	Then $hdev(\alpha, \beta)=D+  [ \bigvee\limits_{\substack{i = 1 \\ \rho_i < +\infty}}^{n} [\tau_i \cdot 1_{ \{ \sigma_i \geq \frac{b-L}{R'-r}r+b \}} - \frac{[\sigma_i -b]^+}{r}, \tau_i \cdot 1_{\{\frac{b-L}{R'-r}r+b\geq \sigma_i \}} + \frac{[\frac{b-L}{R'-r}r+b-\sigma_i]^+}{\rho_i} - \frac{b-L}{R'-r}]^+ 
	    , \bigvee\limits_{\substack{i = 1 \\ \rho_i = +\infty}}^{n} [\tau_i \cdot 1_{ \{ \sigma_i \geq \frac{b-L}{R'-r}r+b \}} - \frac{[\sigma_i -b]^+}{r}, \tau_i \cdot 1_{\{\frac{b-L}{R'-r}r+b\geq \sigma_i \}} - \frac{[\sigma_i-L]^+}{R'}]^+ ]^+$
\end{theorem}

Next, we will consider backlog and output bounds which are important to analyze non-nested tandems and general feedforward networks.

\begin{theorem}[Backlog Bound and Output Arrival Curve of Arrival Curve and Service Curve in \curveclass]\label{theorem:backlog_bound_ouput_ac_simple_ac_plcs_sc}
	Let $\alpha=\gamma_{b,r}$ and $\beta \in$ \curveclass~of the form $\beta=\delta_D \otimes (\bigwedge\limits_{i=1}^n (\delta_{\tau_i}\otimes I_{\sigma_i}\otimes \gamma_{0,\rho_i})\wedge\delta_0)$ with $0 \leq r \leq \min\limits_{i=1,...,n} \rho_i$. 
	Then $vdev(\alpha,\beta)=[b+Dr]\vee[\bigvee\limits_{i=1}^n(b-\sigma_i+(D+\tau_i)r)]$ and an output arrival curve for the flow is $\alpha\oslash\beta(t)=rt+vdev(\alpha,\beta)=rt +[b+Dr]\vee[\bigvee\limits_{i=1}^n(b-\sigma_i+(D+\tau_i)r)]$ for $t\geq 0$.
\end{theorem}

We also need to discuss how to compute the leftover service curves with LUDB++.

\begin{theorem}[Non-Decreasing Lower Bound of Leftover Service in \curveclass]\label{theorem:non_decreasing_lb_left_over_sc_with_shaped_ac_in_plcs}
	Let $\alpha=\gamma_{b,r} \wedge \gamma_{L,R'}$ and $\beta \in$ \curveclass~of the form $\beta=\delta_D \otimes (\bigwedge\limits_{i=1}^n (\delta_{\tau_i}\otimes I_{\sigma_i} \otimes \gamma_{0,\rho_i}) \wedge \delta_0)$ with $L \leq b$ and $r < R', R' \geq \max\limits_{i=1,...,n} \{\rho_i | \rho_i < +\infty \} , r \leq \min\limits_{i=1,...,n} \rho_i$. 
	Then,  $\underline{ (\beta \ominus_{\theta} \alpha) }^\uparrow \in$ \curveclass~for every $\theta \geq 0$ where $\underline f^\uparrow$ for $f \in \F$ is the non-decreasing lower bound of $f$ defined as $\underline{f(t)}^\uparrow:= \inf\limits_{u \geq t} \{ f(u) \}$.
\end{theorem}

\begin{corollary}[Symbolic Non-Decreasing Lower Bound of Leftover Service Curve]\label{cor:symbolic_lb_left_over_sc}
Let $\alpha=\gamma_{b,r}\wedge \gamma_{L,R'}$ and $\beta \in$ \curveclass~of the form $\beta=\delta_D \otimes (\bigwedge\limits_{i=1}^n (\delta_{\tau_i}\otimes I_{\sigma_i} \otimes \gamma_{0,\rho_i}) \wedge \delta_0)$ with $L \leq b$ and $r < R', R' \geq \max\limits_{i=1,...,n} \{ \rho_i | \rho_i < +\infty\} , r \leq \min\limits_{i=1,...,n} \rho_i$.
Then, $\underline{ (\beta \ominus_{\theta} \alpha) }^\uparrow$ for $\rho_i < +\infty$ with $\theta \geq D$ can be described as follows 
\begin{itemize}
	\item $\delta_D \otimes (((\delta_{\tau_i+\frac{[y-\sigma_i]^+}{\rho_i-r}}\otimes I_{[\sigma_i-y]^+}\otimes \gamma_{0,\rho_i-r})\wedge \delta_0) \wedge ((\delta_{\theta - D} \otimes I_0 \otimes \gamma_{0,+\infty})\wedge \delta_0))$ with $D \leq \theta \leq D+\tau_i-\frac{b-L}{R'-r}$ and $y=((D+\tau_i)-\theta)r+b$.
	\item $\delta_D \otimes (((\delta_{\theta-D+\frac{b-L}{R'-r}+\frac{[\frac{b-L}{R'-r}r+b-y]^+}{\rho_i-r}} \otimes I_{[y-(\frac{b-L}{R'-r}r+b)]^+} \otimes \gamma_{0,\rho_i-r})\wedge \delta_0) \wedge ((\delta_{\theta-D} \otimes I_0 \otimes \gamma_{0, +\infty})\wedge \delta_0))$ with $\theta \geq D$ and $\theta \geq D + \tau_i - \frac{b-L}{R'-r}$ and $y=(\theta + \frac{b-L}{R'-r}-(D+\tau_i))\rho_i+\sigma_i$.
\end{itemize}
and for $\rho_i = +\infty$ with $\theta \geq D$ it can be described as  
\begin{itemize}
	\item $\delta_{D} \otimes (((\delta_{\tau_i} \otimes I_0 \otimes \gamma_{0,+\infty})\wedge \delta_0) \wedge ((\delta_{\theta - D} \otimes I_0 \otimes \gamma_{0,+\infty})\wedge \delta_0))$ with $D \leq \theta \leq D+\tau_i$.
	\item $\delta_{D} \otimes (((\delta_{\theta - D} \otimes I_0 \otimes \gamma_{0,+\infty})\wedge \delta_0) \wedge ((\delta_{\theta - D} \otimes I_0 \otimes \gamma_{0,+\infty})\wedge \delta_0))$ with $\theta > D+\tau_i$.
\end{itemize}
\end{corollary}

With the aforementioned results, we are now able to analyze shaped FIFO networks with LUDB++ and provide an example in the following.
\subsection{Example}
In this (sub-)Section, we present the LUDB++ analysis steps for two server topology depicted in Figure \ref{fig:2_flows_2_servers}.
\begin{figure}[H]  
\begin{centering}
  \includegraphics[width=0.5\textwidth]{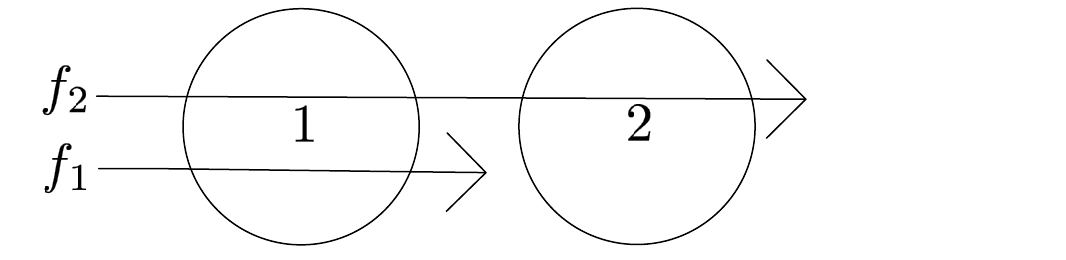}  
  \caption{LUDB++ example for two servers with one crossflow.}
  \label{fig:2_flows_2_servers}
\end{centering}
\end{figure}
For delay bounding, LUDB++ computes $\inf\limits_{\theta\geq0} hdev(\alpha_2,\beta_2 \otimes \underline{(\beta_1 \ominus_{\theta} \alpha_1)}^\uparrow)$ with $\alpha_i \in \gamma_{b_i,r_i}\wedge \gamma_{L_i,R_i'}$ and $\beta_i \in$ \curveclass~with $n=1$. \\
By Corollary \ref{cor:symbolic_lb_left_over_sc}, $\underline{(\beta_1 \ominus_{\theta} \alpha_1)}^\uparrow$ yields: \\
$(I)$ $\delta_{D_1} \otimes ([\delta_{\tau_1+\frac{[y-\sigma_1]^+}{\rho_1-r_1}} \otimes I_{[\sigma_1-y]^+}\otimes \gamma_{0,\rho_1-r_1}]\wedge [\delta_{\theta-D_1}\otimes I_0 \otimes \gamma_{0,+\infty}] \wedge \delta_0)$ with $\theta \geq D_1, \theta \leq D_1 + \tau_1 - \frac{b_1 - L_1}{R_1'-r_1}$ and $y=((D_1+\tau_1)-\theta)r_1+b_1$. \\
$(II)$ $\delta_{D_1} \otimes ([\delta_{\theta-D_1+\frac{b_1-L_1}{R_1'-r_1}+\frac{[\frac{b_1-L_1}{R_1'-r_1}r_1+b_1-y]^+}{\rho_1-r_1}} \otimes I_{[y-(\frac{b_1-L_1}{R_1'-r_1} r_1+b_1)]^+} \otimes \gamma_{0,\rho_1-r_1}] \wedge [\delta_{\theta-D_1} \otimes I_0 \otimes \gamma_{0,+\infty}] \wedge \delta_0)$ with $\theta \geq D_1, \theta \geq D_1 + \tau_1 - \frac{b_1 -L_1}{R_1'-r_1}$ and $y=(\theta+\frac{b_1-L_1}{R_1'-r_1}-(D_1+\tau_1))\rho_1$. \\ \\
For brevity, we only show decomposition $(I)$ in the following. 
For $\beta_2 \otimes \underline{(\beta_1 \ominus_{\theta} \alpha_1)}^\uparrow$, we make use of Theorem \ref{theorem:mslc_closed_under_conv} while already removing duplicate steps in the overall minimum expression as well as some optimization considerations where one of the rates is infinite and obtain the following: \\
$(I)$ $\delta_{D_1+D_2}([\delta_{\tau_2}\otimes I_{\sigma_2} \otimes \gamma_{0,\rho_2}] 
    \wedge [\delta_{\tau_1 + \frac{[y-\sigma_1]^+}{\rho_1-r_1}} \otimes I_{[\sigma_1-y]^+} \otimes \gamma_{0,\rho_1-r_1}] 
    \wedge [\delta_{\tau_2+\tau_1+\frac{[y-\sigma_1]^+}{\rho_1-r_1}} \otimes I_{[\sigma_1-y]^+ +\sigma_2} \otimes \gamma_{0,\rho_2}] \wedge [\delta_{\tau_2+\tau_1+\frac{[y-\sigma_1]^+}{\rho_1-r_1}} \otimes I_{[\sigma_1-y]^+ +\sigma_2} \otimes \gamma_{0,\rho_1-r_1}] 
    \wedge [\delta_{\theta-D_1} \otimes I_0 \otimes \gamma_{0,+\infty}]
    \wedge [\delta_{\tau_2+\theta -D_1} \otimes I_{\sigma_2} \otimes \gamma_{0,\rho_2}] 
    \wedge \delta_0)$
    with $\theta \geq D_1, \theta \leq D_1 + \tau_1 -\frac{b_1-L_1}{R_1'-r_1}$ and $y=((D_1 + \tau_1)-\theta)r_1+b_1$
\\\\
We proceed bottom up as in the LUDB methodology by splitting up the different cases already at this level of the analysis. \\
$(I)$ Let $y=((D_1+\tau_1)-\theta)r_1+b_1$ and constraints $\theta \geq D_1, \theta \leq D_1 + \tau_1 - \frac{b_1 - L_1}{R_1' -r_1}$ \\
$(I.I)$ Additional constraint $y-\sigma_1 \geq 0$ yields $\delta_{D_1+D_2} \otimes ([\delta_{\tau_2} \otimes I_{\sigma_2} \otimes \gamma_{0, \rho_2}] 
\wedge [\delta_{\tau_1 + \frac{y-\sigma_1}{\rho_1 - r_1}} \otimes I_0 \otimes \gamma_{0, \rho_1 - r_1}]
\wedge [\delta_{\tau_2 + \tau_1 + \frac{y-\sigma_1}{\rho_1 - r_1}} \otimes I_{\sigma_2} \otimes \gamma_{0,\rho_2}]
\wedge [\delta_{\tau_2 + \tau_1 + \frac{y-\sigma_1}{\rho_1 - r_1}} \otimes I_{\sigma_2} \otimes \gamma_{0,\rho_1-r_1}]
\wedge [\delta_{\theta-D_1} \otimes I_0 \otimes \gamma_{0,+\infty}] \
\wedge [\delta_{\tau_2+\theta-D_1} \otimes I_{\sigma_2} \otimes \gamma_{0,\rho_2}]\wedge \delta_0)$
\\
$(I.II)$ Additional constraint $y-\sigma_1 \leq 0$ yields $\delta_{D_1 + D_2} \otimes 
([\delta_{\tau_2} \otimes I_{\sigma_2} \otimes \gamma_{0,\rho_2}] 
\wedge [\delta_{\tau_1} \otimes I_{\sigma_1 - y} \otimes \gamma_{0, \rho_1 - r_1}]
\wedge [\delta_{\tau_2+\tau_1} \otimes I_{\sigma_2+\sigma_1-y} \otimes \gamma_{0,\rho_2}]
\wedge [\delta_{\tau_2+\tau_1} \otimes I_{\sigma_2+\sigma_1 -y} \otimes \gamma_{0,\rho_1 -r_1}]
\wedge [\delta_{\theta-D_1} \otimes I_0 \otimes \gamma_{0,+\infty}]
\wedge [\delta_{\tau_2+\theta -D_1} \otimes I_{\sigma_2} \otimes \gamma_{0,\rho_2}]\wedge \delta_0)$
\\\\
Again, for brevity, we only show the decomposition $(I.I)$ in the following. 
Now, we apply Theorem \ref{theorem:delay_bound_shaped_ac_plcs_sc} for a delay bound expression of this composition and obtain \\
$(I.I)$ $D_1+D_2+[\tau_2 \cdot 1_{\{\sigma_2 \geq \frac{b_2-L_2}{R_2'-r_2}r_2+b_2\}} -\frac{[\sigma_2-b_2]^+}{r_2}, 
\tau_2 \cdot 1_{\{\frac{b_2-L_2}{R_2'-r_2}r_2+b_2\geq \sigma_2\}} + \frac{[\frac{b_2-L_2}{R_2'-r_2}r_2+b_2-\sigma_2]^+}{\rho_2}-\frac{b_2-L_2}{R_2'-r_2}]^+ 
\bigvee [(\tau_1+\frac{y-\sigma_1}{\rho_1-r_1}) \cdot 1_{\{0\geq \frac{b_2-L_2}{R_2'-r_2}r_2+b_2 \}} -\frac{[0-b_2]^+}{r_2},
(\tau_1+\frac{y-\sigma_1}{\rho_1-r_1})\cdot 1_{\{\frac{b_2-L_2}{R_2'-r_2}r_2+b_2 \geq 0\}} + \frac{[\frac{b_2-L_2}{R_2'-r_2}r_2+b_2-0]^+}{\rho_1-r_1}-\frac{b_2-L_2}{R_2'-r_2}]^+
\bigvee [(\tau_2+\tau_1+\frac{y-\sigma_1}{\rho_1-r_1})\cdot 1_{\{\sigma_2 \geq \frac{b_2-L_2}{R_2'-r_2}r_2+b_2\}}-\frac{[\sigma_2-b_2]^+}{r_2},
(\tau_2+\tau_1+\frac{y-\sigma_1}{\rho_1-r_1})\cdot 1_{\{\frac{b_2-L_2}{R_2'-r_2}r_2+b_2\geq \sigma_2\}} +\frac{[\frac{b_2-L_2}{R_2'-r_2}r_2+b_2-\sigma_2]^+}{\rho_2}-\frac{b_2-L_2}{R_2'-r_2}]^+
\bigvee [(\tau_2+\tau_1+\frac{y-\sigma_1}{\rho_1-r_1})\cdot 1_{\{\sigma_2 \geq \frac{b_2-L_2}{R_2'-r_2}r_2+b_2\}}-\frac{[\sigma_2-b_2]^+}{r_2},
(\tau_2+\tau_1+\frac{y-\sigma_1}{\rho_1-r_1})\cdot 1_{\{\frac{b_2-L_2}{R_2'-r_2}r_2+b_2 \geq \sigma_2\}} +\frac{[\frac{b_2-L_2}{R_2'-r_2}r_2+b_2-\sigma_2]^+}{\rho_1-r_1}-\frac{b_2-L_2}{R_2'-r_2}]^+
\bigvee [(\theta-D_1) \cdot 1_{\{0 \geq \frac{b_2-L_2}{R_2'-r_2}r_2+b_2\}}-\frac{[0-b_2]^+}{r_2},
(\theta-D_1)\cdot 1_{\{\frac{b_2-L_2}{R_2'-r_2}r_2+b_2\geq 0\}} -\frac{[0-L_2]^+}{R_2'}]^+
\bigvee [(\tau_2+\theta-D_1) \cdot 1_{\{\sigma_2 \geq \frac{b_2-L_2}{R_2'-r_2}r_2+b_2 \}} -\frac{[\sigma_2-b_2]^+}{r_2},
(\tau_2+\theta -D_1)\cdot 1_{\{\frac{b_2-L_2}{R_2'-r_2}r_2+b_2 \geq \sigma_2 \}} + \frac{[\frac{b_2-L_2}{R_2'-r_2}r_2+b_2-\sigma_2]^+}{\rho_2} -\frac{b_2-L_2}{R_2'-r_2}]^+$
with $y=((D_1+\tau_1)-\theta)r_1+b_1$ and constraints $\theta \geq D_1, \theta \leq D_1+\tau_1-\frac{b_1-L_1}{R_1'-r_1}, y-\sigma_1 \geq 0$.
\\\\
The decomposition $(I.I)$ with further constraint decompositions $\frac{b_2-L_2}{R_2'-r_2}r_2+b_2 \geq 0, \frac{b_2-L_2}{R_2'-r_2}r_2+b_2 \geq \sigma_2$ and where the second element in $[]^+$ is the maximum element yields the following LP 
\\
\\
$\text{Minimize } D_1 +D_2 + \tau_2 + \frac{\frac{b_2-L_2}{R_2'-r_2}r_2+b_2 -\sigma_2}{\rho_2} - \frac{b_2-L_2}{R_2'-r_2} $\\\\
$\text{such that} $\\\\
$\frac{b_2-L_2}{R_2'-r_2}r_2+b_2 \geq 0 $\\
$\frac{b_2-L_2}{R_2'-r_2}r_2+b_2 \geq \sigma_2 $\\
$\tau_2 + \frac{\frac{b_2-L_2}{R_2'-r_2}r_2+b_2 -\sigma_2}{\rho_2} - \frac{b_2-L_2}{R_2'-r_2} \geq 0 $\\
$\tau_2 + \frac{\frac{b_2-L_2}{R_2'-r_2}r_2+b_2 -\sigma_2}{\rho_2} - \frac{b_2-L_2}{R_2'-r_2} \geq 0 $\\ 
$\tau_2 + \frac{\frac{b_2-L_2}{R_2'-r_2}r_2+b_2 -\sigma_2}{\rho_2} - \frac{b_2-L_2}{R_2'-r_2} \geq 0$\\ 
$\tau_2 + \frac{\frac{b_2-L_2}{R_2'-r_2}r_2+b_2 -\sigma_2}{\rho_2} - \frac{b_2-L_2}{R_2'-r_2} \geq \tau_1 +\frac{y-\sigma_1}{\rho_1-r_1}+\frac{\frac{b_2-L_2}{R_2'-r_2}r_2+b_2}{\rho_1-r_1}-$\\
$~~~~~~~~~~~~~~~~~~~~~~~~~~~~~~~~~~~~~~~\frac{b_2-L_2}{R_2'-r_2}$\\
$\tau_2 + \frac{\frac{b_2-L_2}{R_2'-r_2}r_2+b_2 -\sigma_2}{\rho_2} - \frac{b_2-L_2}{R_2'-r_2} \geq 0$\\ 
$\tau_2 + \frac{\frac{b_2-L_2}{R_2'-r_2}r_2+b_2 -\sigma_2}{\rho_2} - \frac{b_2-L_2}{R_2'-r_2} \geq \tau_2 + \tau_1 +\frac{y-\sigma_1}{\rho_1-r_1}+$\\
$~~~~~~~~~~~~~~~~~~~~~~~~~~~~~~~~~~~~~~~\frac{\frac{b_2-L_2}{R_2'-r_2}r_2+b_2-\sigma_2}{\rho_2}-\frac{b_2-L_2}{R_2'-r_2}$\\ 
$\tau_2 + \frac{\frac{b_2-L_2}{R_2'-r_2}r_2+b_2 -\sigma_2}{\rho_2} - \frac{b_2-L_2}{R_2'-r_2} \geq 0$\\ 
$\tau_2 + \frac{\frac{b_2-L_2}{R_2'-r_2}r_2+b_2 -\sigma_2}{\rho_2} - \frac{b_2-L_2}{R_2'-r_2} \geq \tau_2 + \tau_1 +\frac{y-\sigma_1}{\rho_1-r_1}+$\\
$~~~~~~~~~~~~~~~~~~~~~~~~~~~~~~~~~~~~~~~\frac{\frac{b_2-L_2}{R_2'-r_2}r_2+b_2-\sigma_2}{\rho_1-r_1}-\frac{b_2-L_2}{R_2'-r_2}$\\ 
$\tau_2 + \frac{\frac{b_2-L_2}{R_2'-r_2}r_2+b_2 -\sigma_2}{\rho_2} - \frac{b_2-L_2}{R_2'-r_2} \geq 0$\\ 
$\tau_2 + \frac{\frac{b_2-L_2}{R_2'-r_2}r_2+b_2 -\sigma_2}{\rho_2} - \frac{b_2-L_2}{R_2'-r_2} \geq \theta -D_1$\\ 
$\tau_2 + \frac{\frac{b_2-L_2}{R_2'-r_2}r_2+b_2 -\sigma_2}{\rho_2} - \frac{b_2-L_2}{R_2'-r_2} \geq 0$\\ 
$\tau_2 + \frac{\frac{b_2-L_2}{R_2'-r_2}r_2+b_2 -\sigma_2}{\rho_2} - \frac{b_2-L_2}{R_2'-r_2} \geq \tau_2 + \theta  - D_1 + \frac{\frac{b_2-L_2}{R_2'-r_2}r_2+b_2 -\sigma_2}{\rho_2} - $\\
$~~~~~~~~~~~~~~~~~~~~~~~~~~~~~~~~~~~~~ \frac{b_2-L_2}{R_2'-r_2}$\\ 
$\theta \geq D_1 $\\
$\theta \leq D_1 + \tau_1 -\frac{b_1-L_1}{R_1'-r_1} $\\
$y -\sigma_1 \geq 0$ \\\\
with $y=((D_1+\tau_1)-\theta)r_1+b_1$.\\

Each LP will be solved individually and the minimum objective value over all these LPs will be then the LUDB++ result.
The result is formally equal to $\inf\limits_{\theta\geq0} hdev(\alpha_2,\beta_2 \otimes \underline{(\beta_1 \ominus_{\theta} \alpha_1)}^\uparrow)$.
\\
The procedure for backlog bounds is done in an analogous way.

\section{Evaluation}\label{sec:eval}
We begin the evaluation Section by noting the hard- and software with which we obtained our data\footnote{Released after acceptance.}.
\subsection{Hardware and Software}

\begin{itemize}
\item 	Experiments were performed on a MacBook Pro with macOS Sequoia 15.6.1 (24G90), Apple M1 Chip (2020), 8GB RAM.
\item LUDB++ data were obtained with the code\footnote{Released after acceptance.}, version 1.0.0, C++17, cmake 3.30, CLion 2024.3.5, CPLEX V22.1.1.0. Release build with compiler flags -O3 -march=armv8.5-a -funroll-loops -flto=thin. %https://github.com/alexscheffler/LUDBPlusPlus
\item LUDB-FF, SFA-FIFO data were obtained with the code\footnote{\url{https://github.com/alexscheffler/DNC}}, commit 5bfcdeea705130f719304257b3530dac0b7b36ed, Java 19 SDK, IntelliJ IDEA Ultimate 2024.2.4, CPLEX V22.1.1.0. 
\item ELP, PLP, PLPP, TFA++, AdmTfa data were obtained with the code\footnote{\url{https://github.com/anne-bou/panco}}, commit 2c923fc8618d73bbc589d95450ec530b37b58a1e, Python 3.9.6, PyCharm 2024.1.4 (Professional Edition), lp\_solve version 5.5.2.11
%\item The dataset can be found at  %\url{https://github.com/alexscheffler/dataset-TBD2026}
\end{itemize}

\subsection{Networks and Results}
We evaluate three different kind of network topologies in the following. 
For all of them the following holds.
Each server has service curve $\beta_{R,T}=\beta_{R,1}$ and each flow has arrival curve $\alpha = \gamma_{b,r}=\gamma_{1,1}$. 
Moreover, we assume that each flow is shaped at the ingress location and each server has an (output) shaper with $\gamma_{L,R'}=\gamma_{0.5, R'}$. 
We vary $R$ and $R'$ to achieve different utilizations and different shaper rate to service rate ratios. %$R'=R, 2R, 3R$
\\
\\
We use the following metric to compare delay bounds: $\text{delay}_{A,B} =  \frac{\text{delay}(A) - \text{delay}(B)}{\text{delay}(B)}$, say relative delay bound (of A) to B.
A negative value hence signifies that A's bound is smaller than B's.

\subsubsection{One Hop Persistent Tandems}
These nested tandems have a foi that crosses all servers and for each server there is exactly one crossflow which only crosses exactly that one server.

\begin{figure}[H]  
\begin{centering}
  \includegraphics[width=0.5\textwidth]{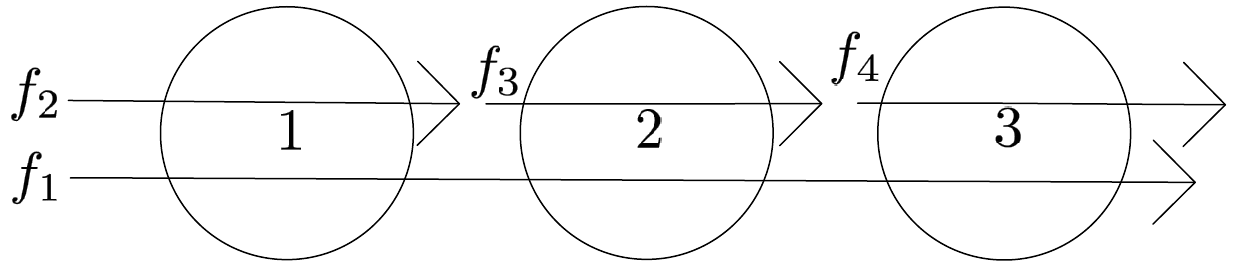}  
  \caption{One hop persistent tandem with $N=3$.}
\end{centering}
\end{figure}

The utilization for this topology is $u=\frac{2 \cdot r}{R}$.
\\
\begin{figure*}[t!]
    \subfloat[$u = 50\%, R'=R$]{{\includegraphics[width=6cm]{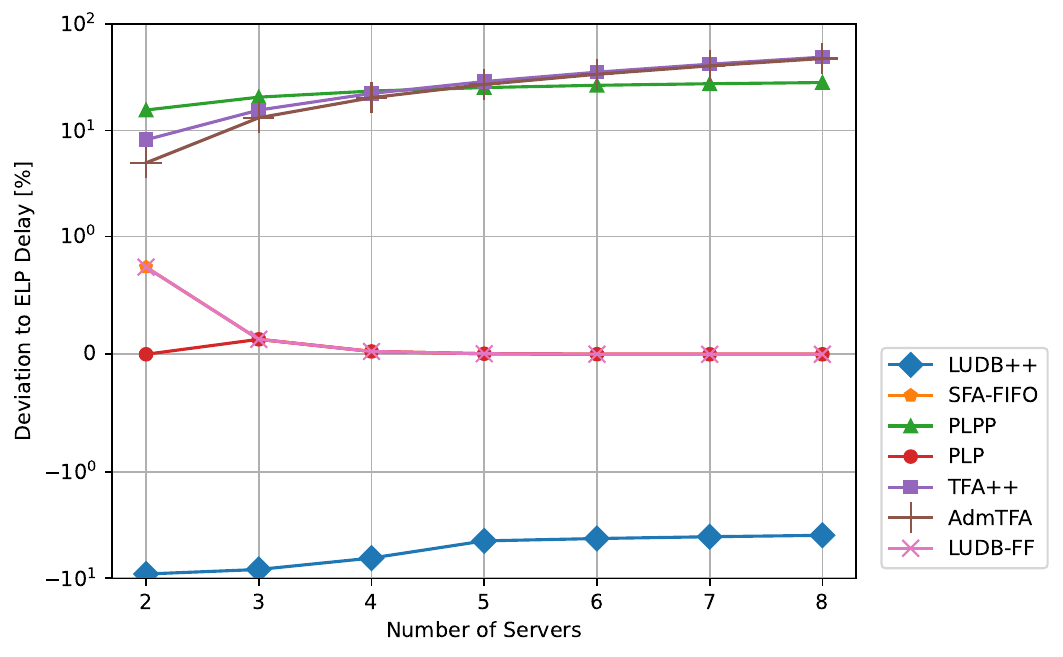} } }
    \subfloat[$u = 50\%, R'=2R$]{{\includegraphics[width=6cm]{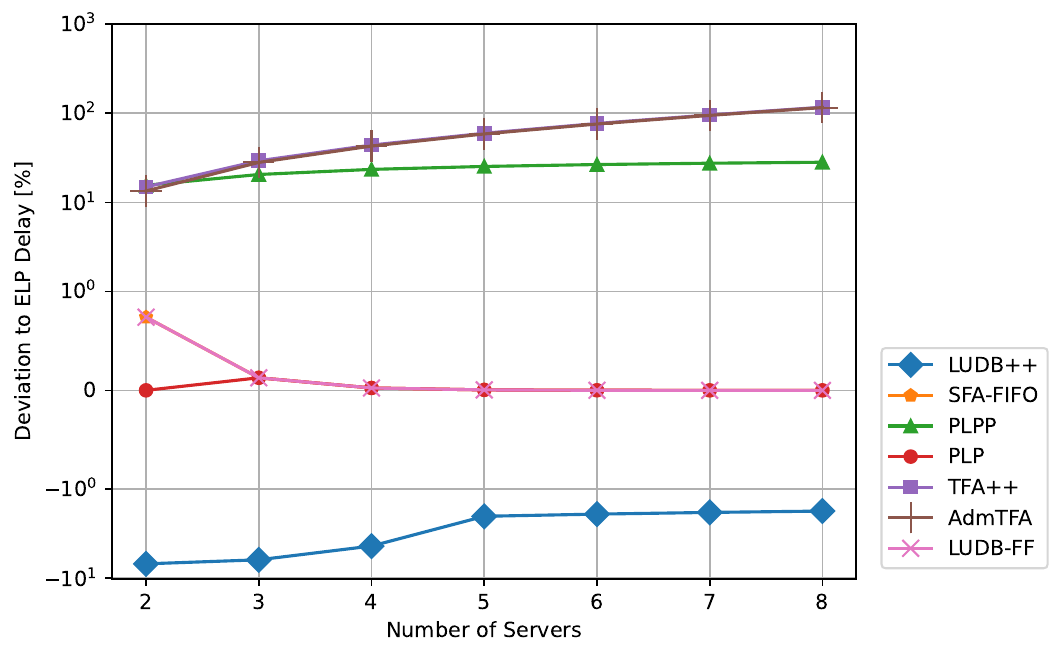} }} 
    \subfloat[$u = 50\%, R'=3R$]{{\includegraphics[width=6cm]{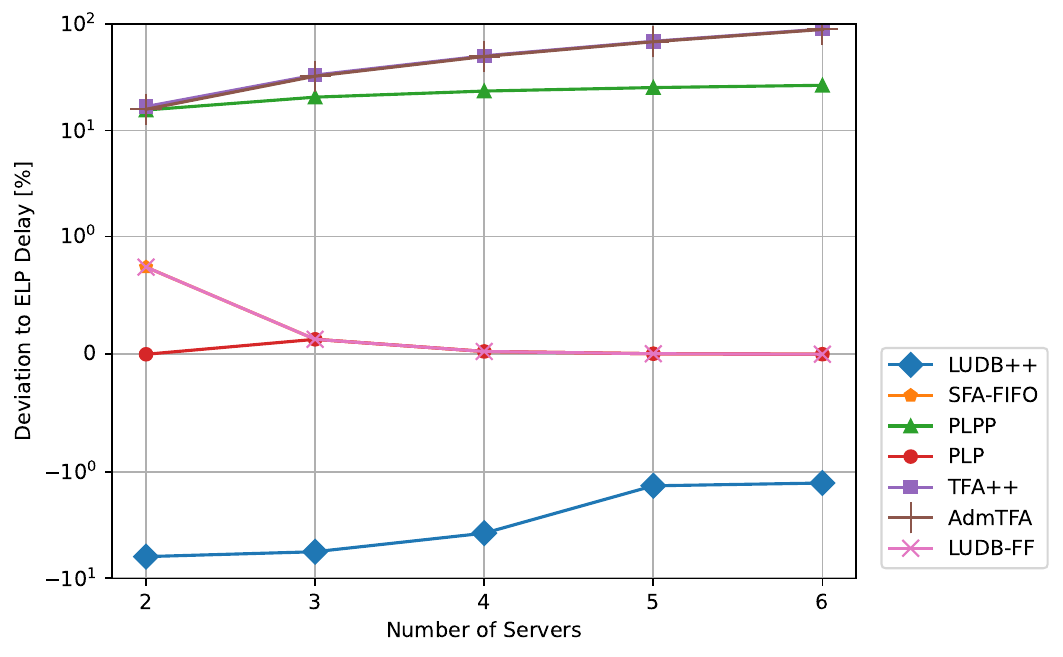} }}  \\
    \subfloat[$u = 75\%, R'=R$]{{\includegraphics[width=6cm]{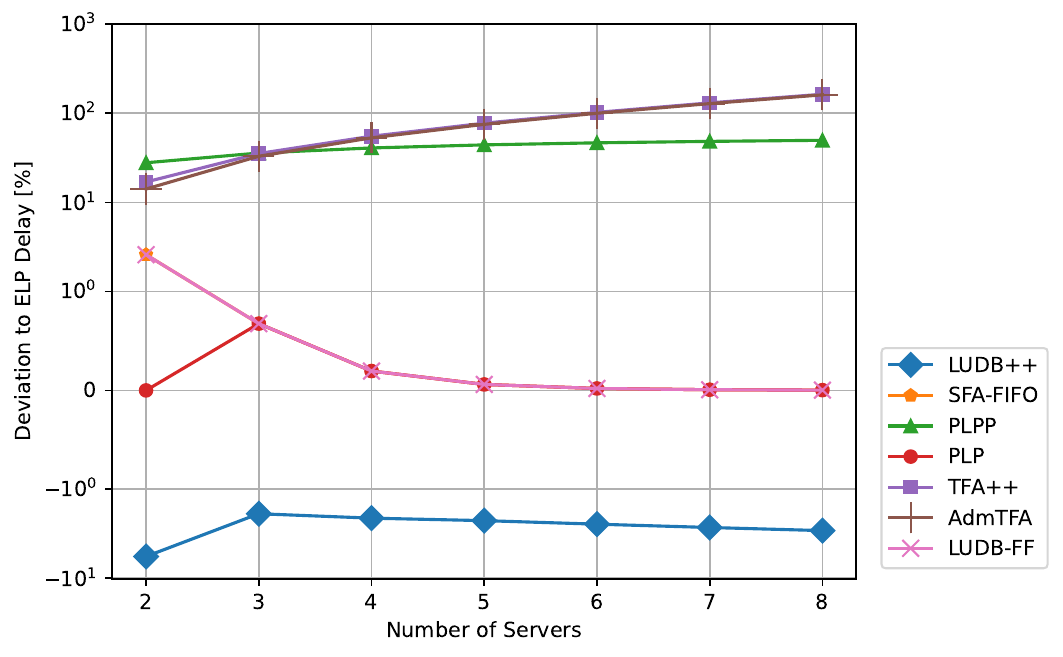} } } 
    \subfloat[$u = 75\%, R'=2R$]{{\includegraphics[width=6cm]{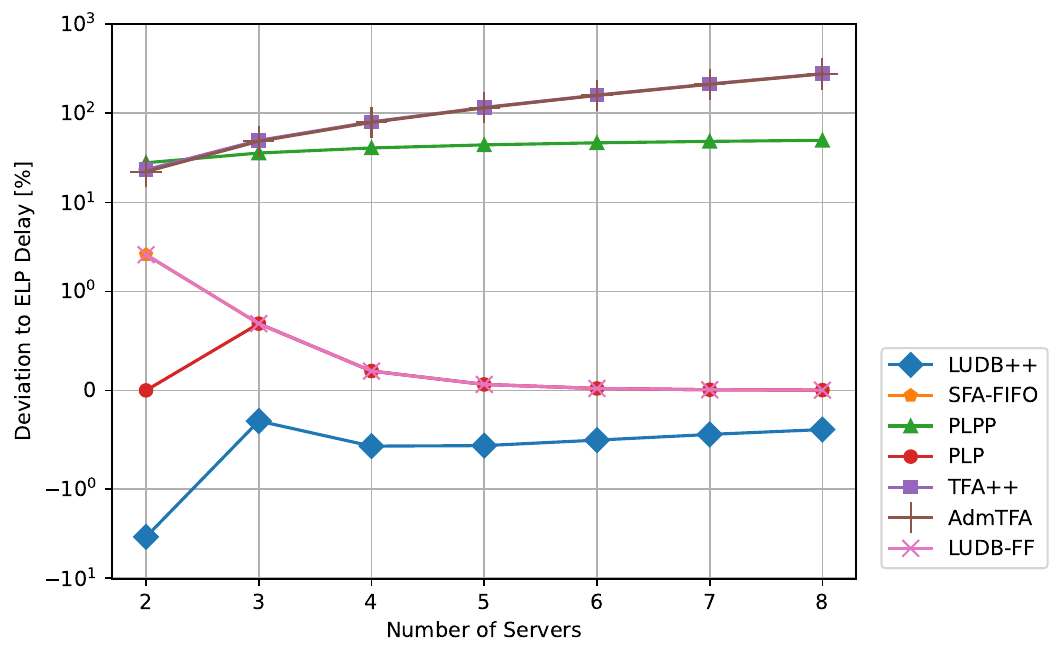} }}
    \subfloat[$u = 75\%, R'=3R$]{{\includegraphics[width=6cm]{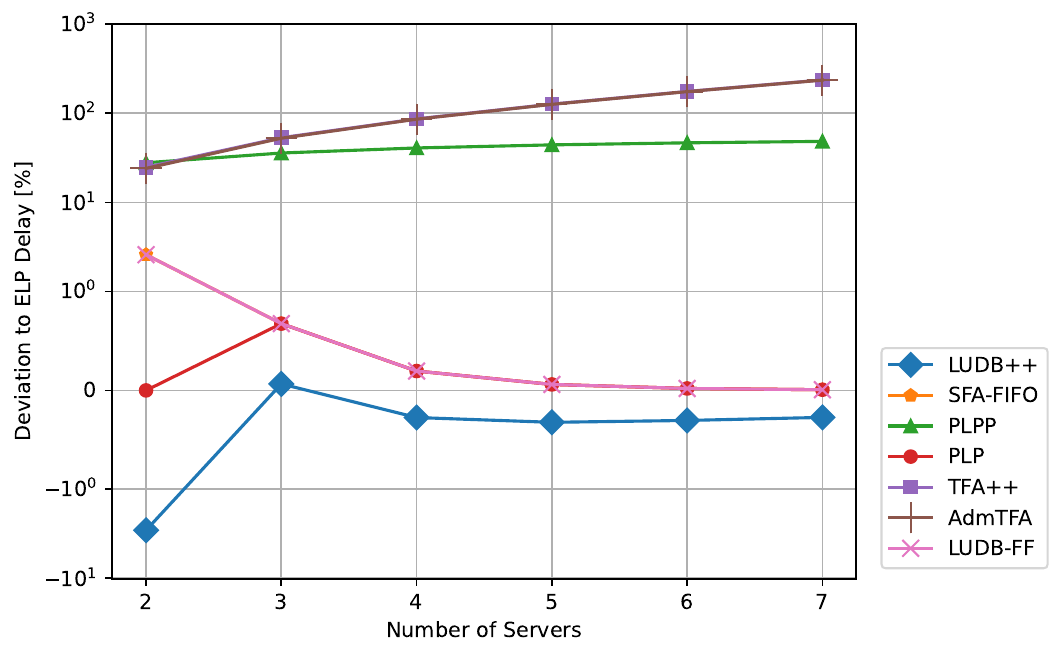} }}     \\
    \subfloat[$u = 100\%, R'=R$]{{\includegraphics[width=6cm]{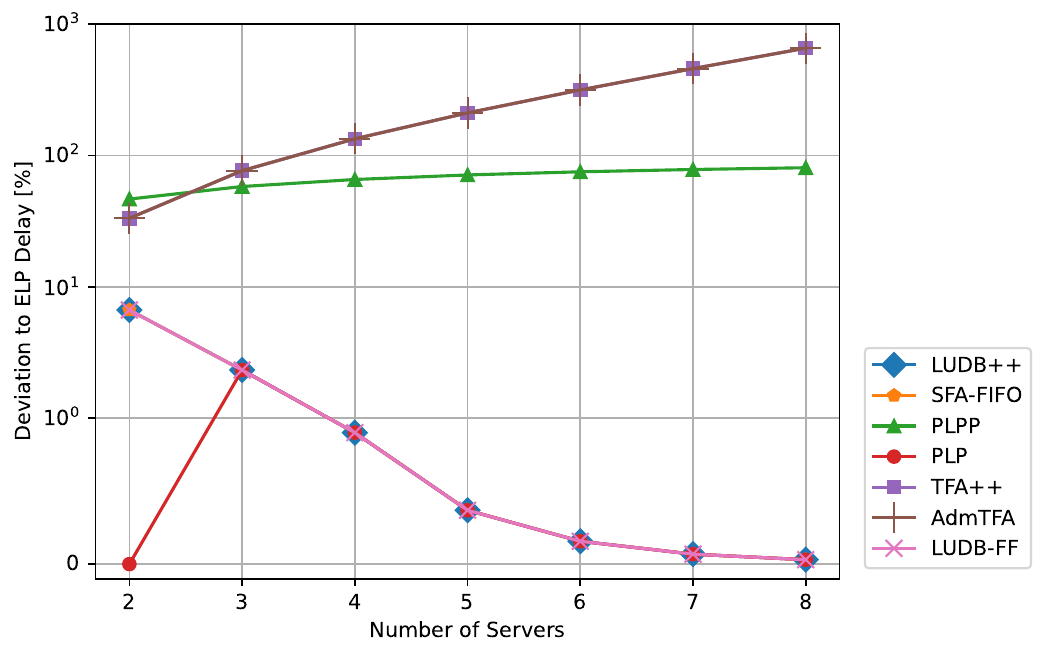} } } 
    \subfloat[$u = 100\%, R'=2R$]{{\includegraphics[width=6cm]{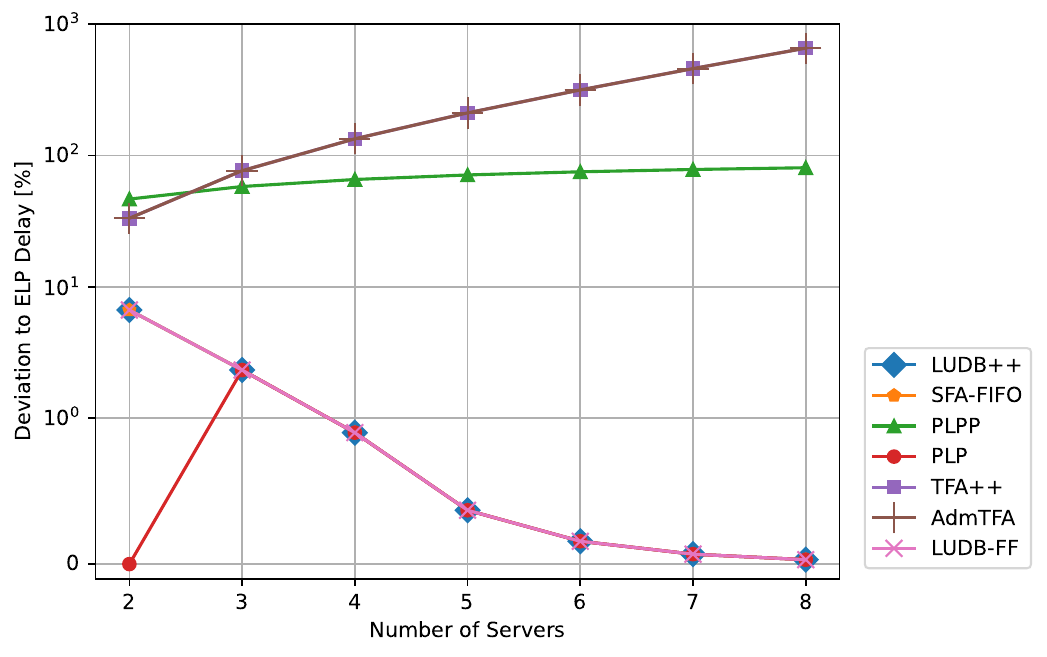} }} 
    \subfloat[$u = 100\%, R'=3R$]{{\includegraphics[width=6cm]{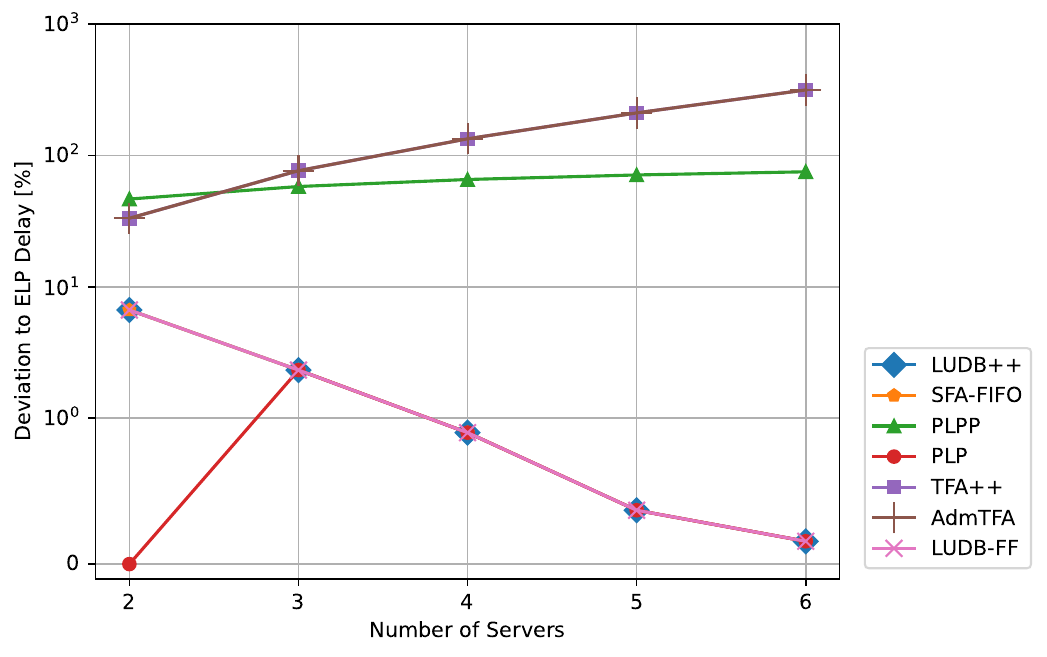} }}
    \caption{Relative delays to ELP for different shaper configurations, different utilizations and different $R'$ to $R$ ratios for one hop persistent tandems.}
    \label{fig:eval_one_hop_persistent}
\end{figure*}

From 63 settings, ELP failed for 5 settings where panco terminated with an error. 
This is the reason why some plots end at 6 or 7 servers instead of 8.
LUDB++ beats ELP for the majority of the settings, namely in 38 settings from 58 settings. 
The minimal, maximal and average deviation of LUDB++ to ELP is $-9.13\%$, $6.66\%$ and $-1.42\%$ respectively.
ELP shows no influence on $R'$ (not shown in the figures) for this topology. 
%In general, ELP is link shaping aware and input shaping non-aware. 
%Same goes for PLP and PLPP and both show no influence of $R'$ for this topology.
Hence, the plots precisely depict the influence of  $R'$ on LUDB++'s delay bounds. 
For a fixed utilization, a decrease of $R'$ yields a larger absolute deviation of LUDB++ to ELP.
This is because a smaller $R'$ has a larger impact on the shaped input arrival curve which in turn decreases the LUDB++ delay bounds.
%and thus enlarging the gap to ELP. 
For the corner case $u=100\%$, note that $R'$ no longer influences LUDB++'s delay bounds such that they coincide with the bounds from LUDB-FF and are worse than ELP.
%Moreover, decreasing the utilization increases the delay gap between LUDB++ and ELP. 
In most cases, increasing the number of servers decreases the gap between ELP and LUDB++ because the shaping effect at the entry of the network that LUDB++ considers in contrast to ELP has a smaller weight. 
As expected, LUDB++ delay bounds are less than or equal to LUDB-FF as the latter does not take shaping into account.
Interestingly, SFA-FIFO  delivers the same bounds as LUDB-FF. 
This can be explained as follows.
There is no cutting for this topology, so LUDB-FF and SFA-FIFO are very similar with the only difference in the optimization of $\theta$ variables -- for this topology and settings they apparently coincide.
PLP delivers the same bounds as LUDB-FF except for 2 servers where it yields a better bound which additionally coincides with ELP. 
%(That PLP=ELP for 2 servers is across all topologies.)
The group which delivers inferior bounds are: AdmTFA, PLPP and TFA++ and the gap to ELP for these methods tends to increase with increasing number of servers.
AdmTFA delivers smaller or equal bounds compared to TFA++ and for small number of servers also than PLPP.

\subsubsection{Sinktree Tandems}
These nested tandems have a foi that crosses all $N$ servers and $N-1$ crossflows.  
Each crossflow starts at a unique position in the tandem (earliest at the second server) and ends in the sink of the tandem.
For an example with $N=3$, refer to Figure \ref{fig:eval_3_servers_sink_tree_tandem}.
The (maximal) utilization for this topology is $u=\frac{N \cdot r}{R}$.
\\
\begin{figure*}[t!]
    \subfloat[$u=50\%, R'=R$]{{\includegraphics[width=6cm]{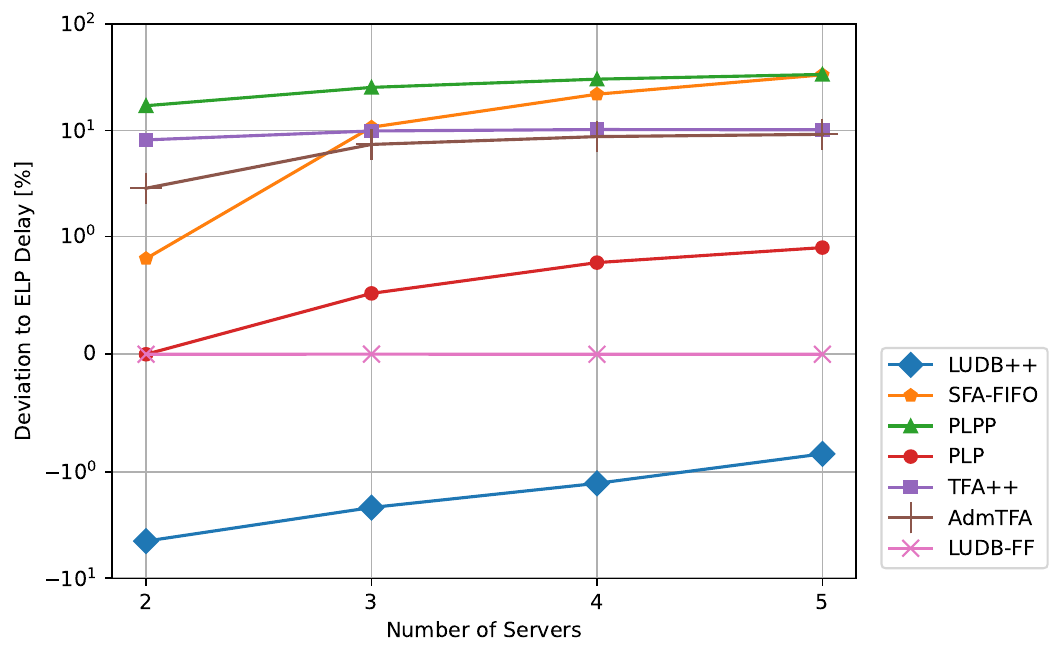} } }
    \subfloat[$u=50\%, R'=2R$]{{\includegraphics[width=6cm]{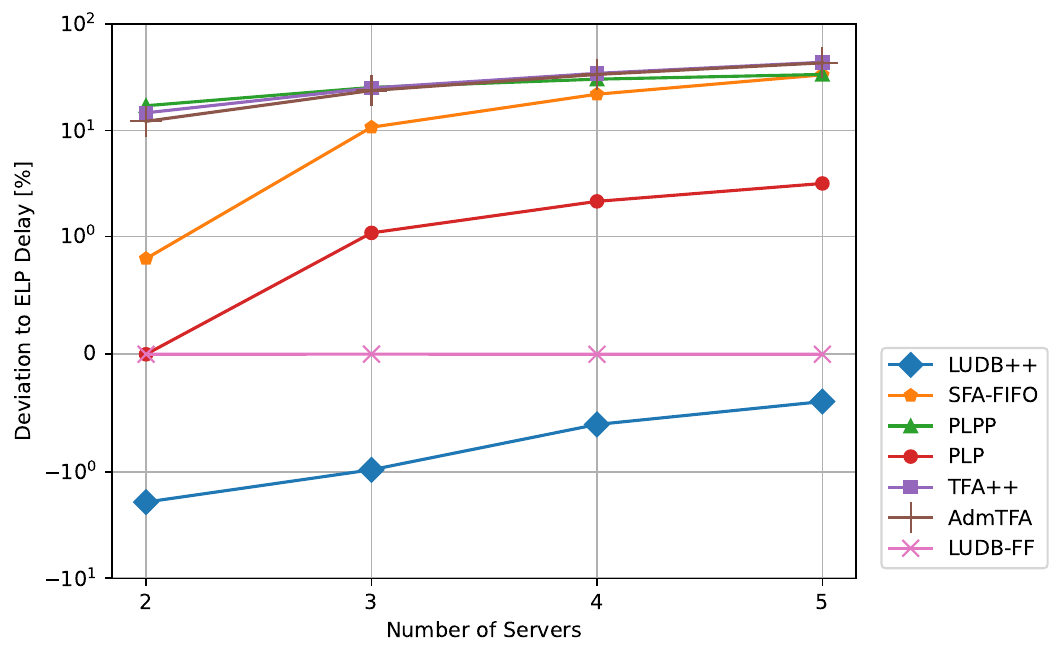} }}
    \subfloat[$u=50\%, R'=3R$]{{\includegraphics[width=6cm]{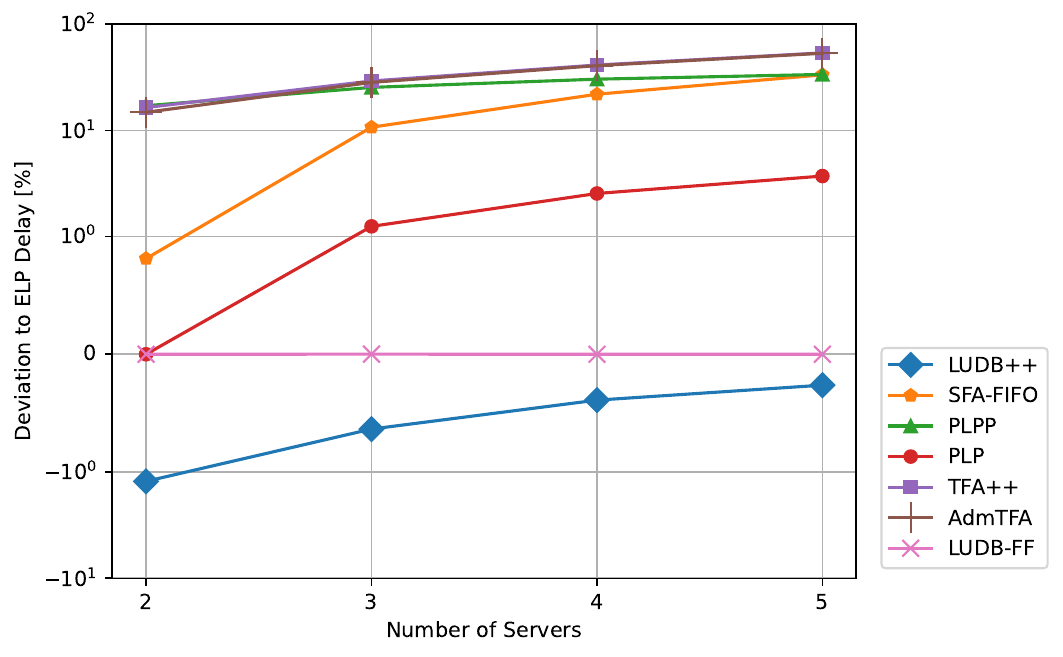} }}
    \\
    \subfloat[$u=75\%, R'=R$]{{\includegraphics[width=6cm]{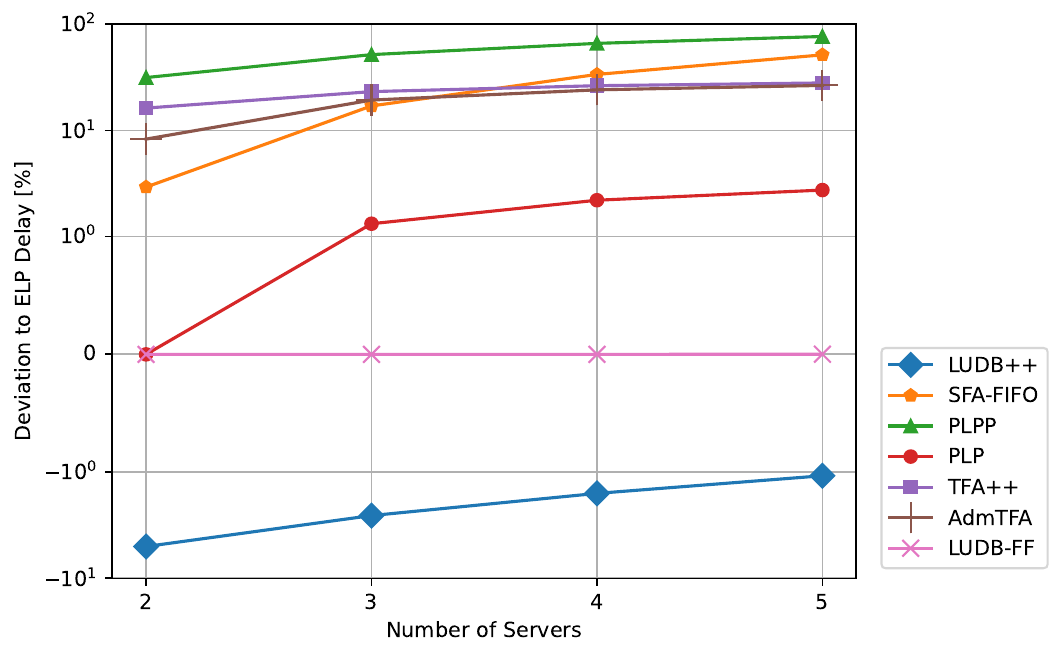} } }
    \subfloat[$u=75\%, R'=2R$]{{\includegraphics[width=6cm]{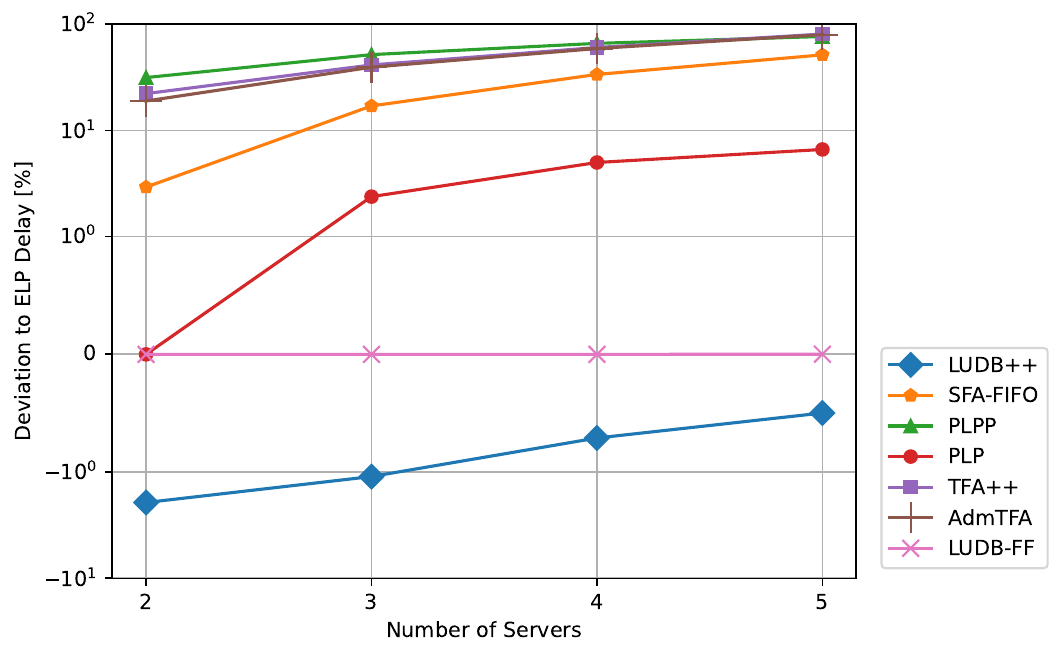} }}
    \subfloat[$u=75\%, R'=3R$]{{\includegraphics[width=6cm]{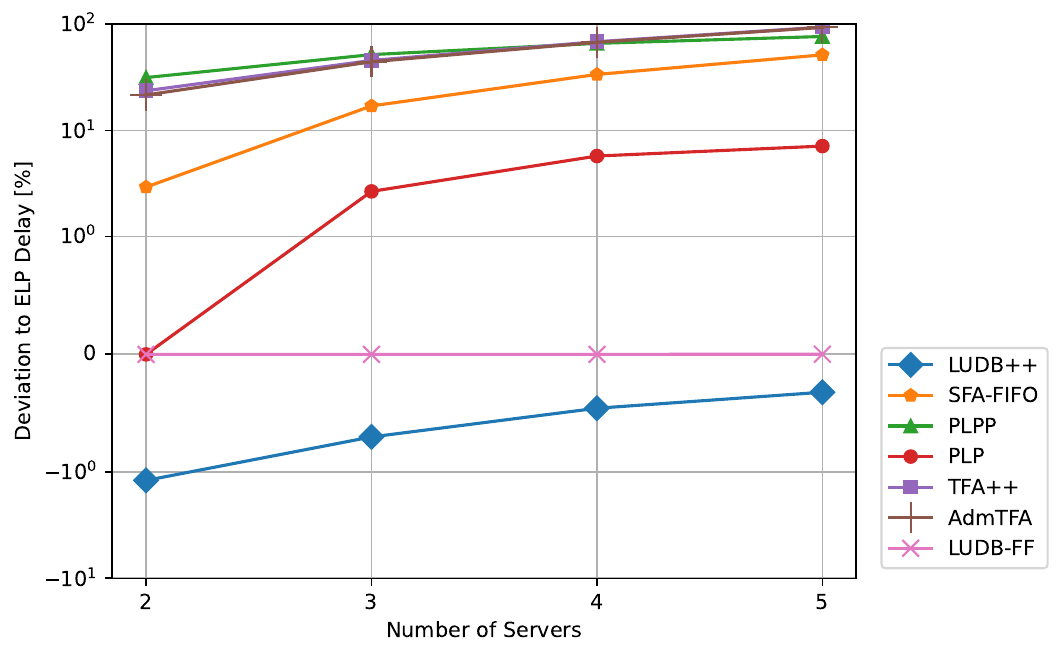} }}
    \\
    \subfloat[$u=100\%, R'=R$]{{\includegraphics[width=6cm]{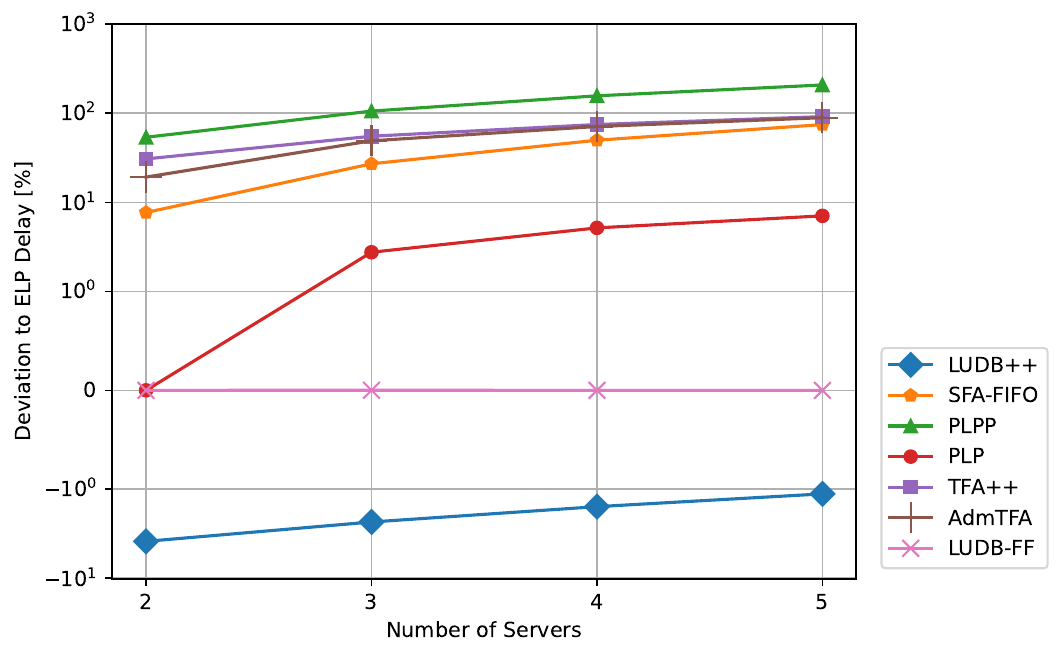} } }
    \subfloat[$u=100\%, R'=2R$]{{\includegraphics[width=6cm]{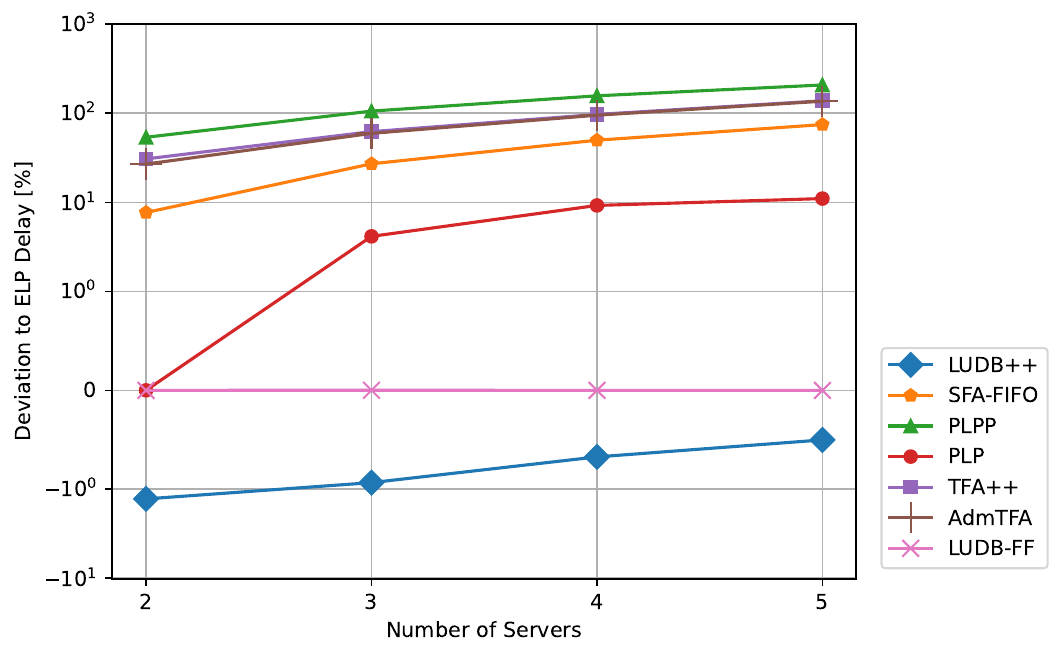} }}
    \subfloat[$u=100\%, R'=3R$]{{\includegraphics[width=6cm]{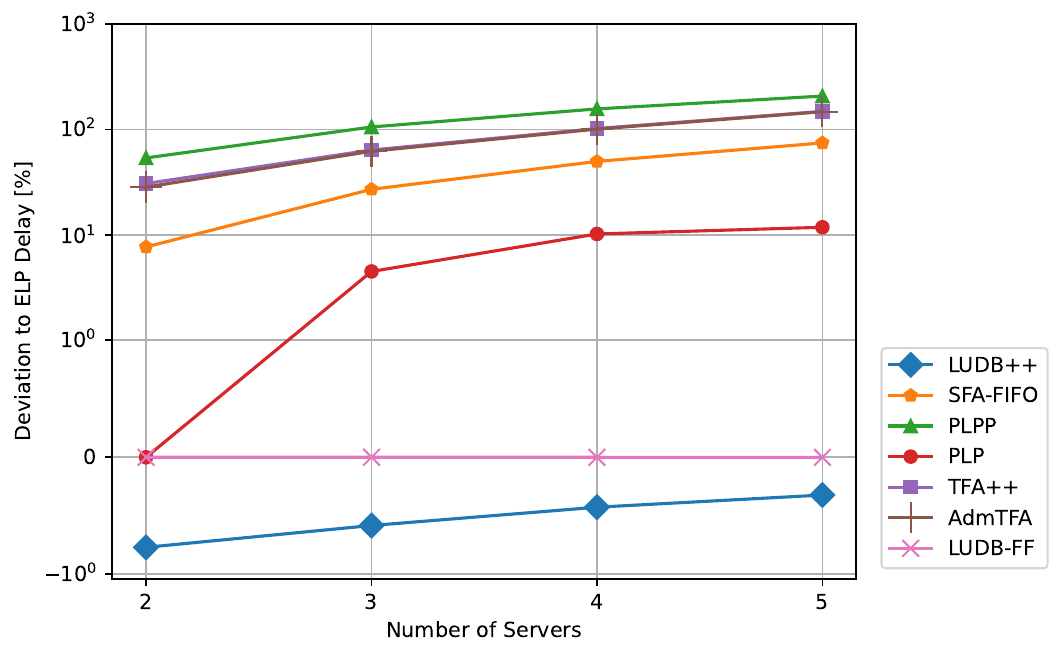} }}
    \caption{Relative delays to ELP for different shaper configurations, different utilizations and different $R'$ to $R$ ratios for sinktree tandems.}
    \label{fig:eval_sink_tree_tandem}
\end{figure*}

LUDB++ beats ELP for all of the 36 settings. 
The minimal, maximal and average deviation of LUDB++ to ELP is $-5.03\%$, $-0.26\%$ and $-1.29\%$ respectively.
ELP shows no influence on $R'$ (not shown in the figures) for this topology. 
%In general, ELP is link shaping aware and input shaping non-aware. 
%The same goes for PLP and PLPP. 
PLPP also shows no influence on $R'$.
PLP is $R'$-sensitive for number of servers at least 3 where an increase of $R'$ increases the delay bound.
For a reasoning on the influence of $R'$ for a fixed utilization on the relative LUDB++ delay bounds, refer to the previous (sub-)Section (one hop tandem topology).
%Because of the previous statement regarding ELP, the plots depict that LUDB++ has an influence on $R'$. 
%For a fixed utilization, a decrease of $R'$ yields a larger absolute deviation of LUDB++ to ELP.  
%A smaller $R'$ has a larger impact on the shaped input arrival curve, decreasing the LUDB++ delay bounds and thus enlarging the gap to ELP. 
	%\item Decreasing the utilization increases the delay gap between LUDB++ and ELP. ?? Does not hold here
In all cases for this topology, increasing the number of servers decreases the gap between ELP and LUDB++.
For the reasoning behind this, refer again to the previous (sub-)Section (one hop tandem topology).
%LUDB++ delay bounds are less than LUDB-FF as expected as the latter does not take shaping into account.
Interestingly, LUDB-FF delivers the same bound as ELP. 
Since ELP does not have an influence on $R'$ here and LUDB-FF is tight for sinktrees (without shaping assumptions), both methods deliver the same bounds\footnote{as formally, for a sinktree topology, LUDB-FF achieves the same bound as the MILP from which the ELP stems from but apparently the latter two coincide for this topology as well.}.
Moreover, PLP delivers the same bound as LUDB-FF for this topology for number of servers 2 and worse bounds for larger number of servers. 
The relative delay gap of PLP to ELP increases with increasing number of servers for this topology.
SFA-FIFO delivers worse bounds compared to PLP and in some scenarios is beaten by TFA++ and AdmTFA, especially for a low $R'$ to $R$ ratio, mid to low utilization and from 3 servers upwards.
PLPP delivers the worst bounds in almost all scenarios for this topology.

\subsubsection{Tree}
As an example for the analysis of general feed-forward networks with this method, the last topology of our evaluation has a tree structure. 
The foi crosses servers on a main branch.
Each crossflow comes from a side branch and joins the foi on the main branch until the last server.
The interference pattern on the main branch is of type sink tree while the pattern on the side branch is of type one hop persistent. 
Figure \ref{fig:3_servers_tree} depicts this topology where the foi crosses 3 servers.
\\ 
\\ 
Note that the side branch objective minimizes for the backlog bound as we are interested in the output arrival curve for the crossflow at the location where it joins the foi.
After computing it, we compute the output arrival curve according to Theorem \ref{theorem:backlog_bound_ouput_ac_simple_ac_plcs_sc} and then shape it with the shaping curve $\gamma_{L,R'}$.

\begin{figure}[H]  
\begin{centering}
  \includegraphics[width=0.5\textwidth]{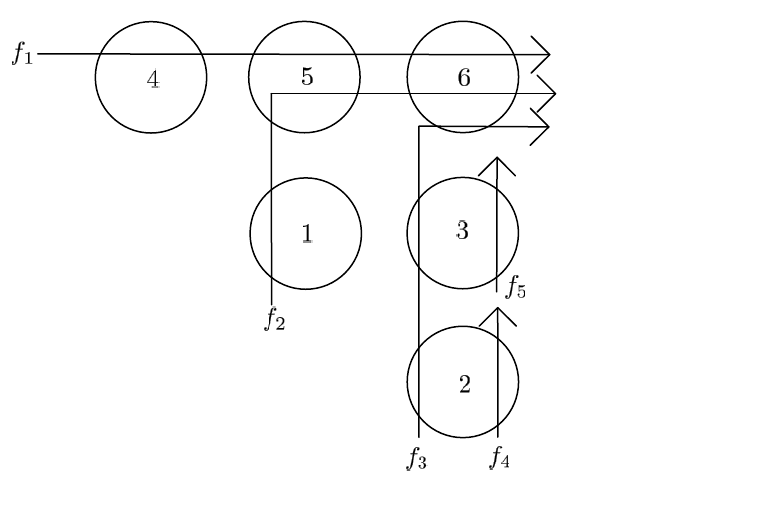}  
  \caption{Tree with $N=3$.}
  \label{fig:3_servers_tree}
\end{centering}
\end{figure}

The (maximal) utilization for this topology is $u=\frac{N \cdot r}{R}$.
\\
\begin{figure*}[t!]
    \subfloat[$u=50\%, R'=R$]{{\includegraphics[width=6cm]{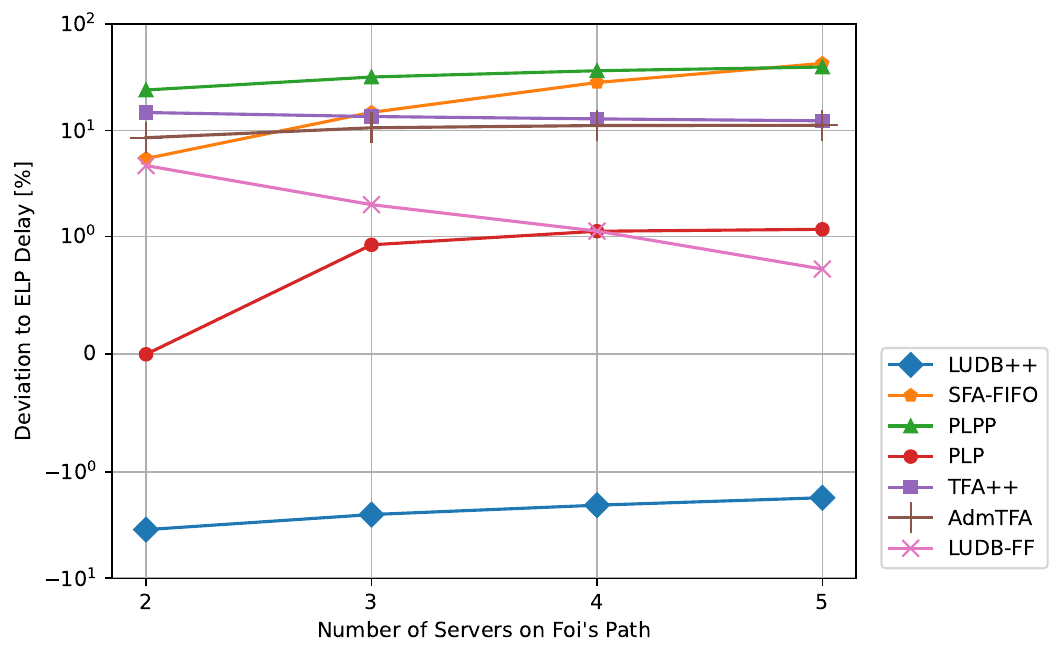} } }
    \subfloat[$u=50\%, R'=2R$]{{\includegraphics[width=6cm]{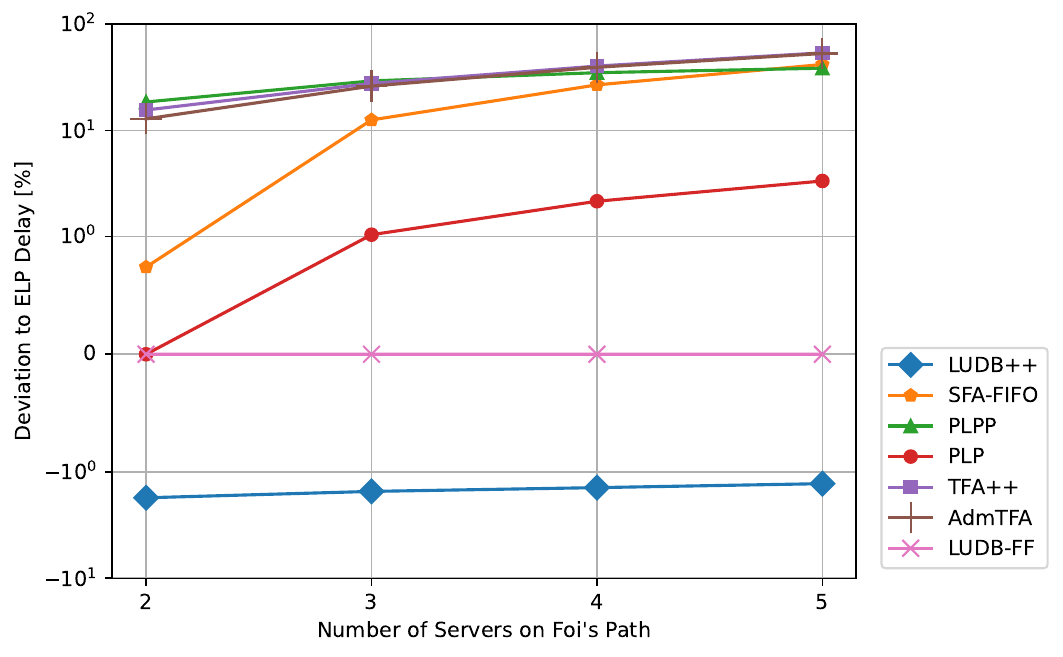} }}
    \subfloat[$u=50\%, R'=3R$]{{\includegraphics[width=6cm]{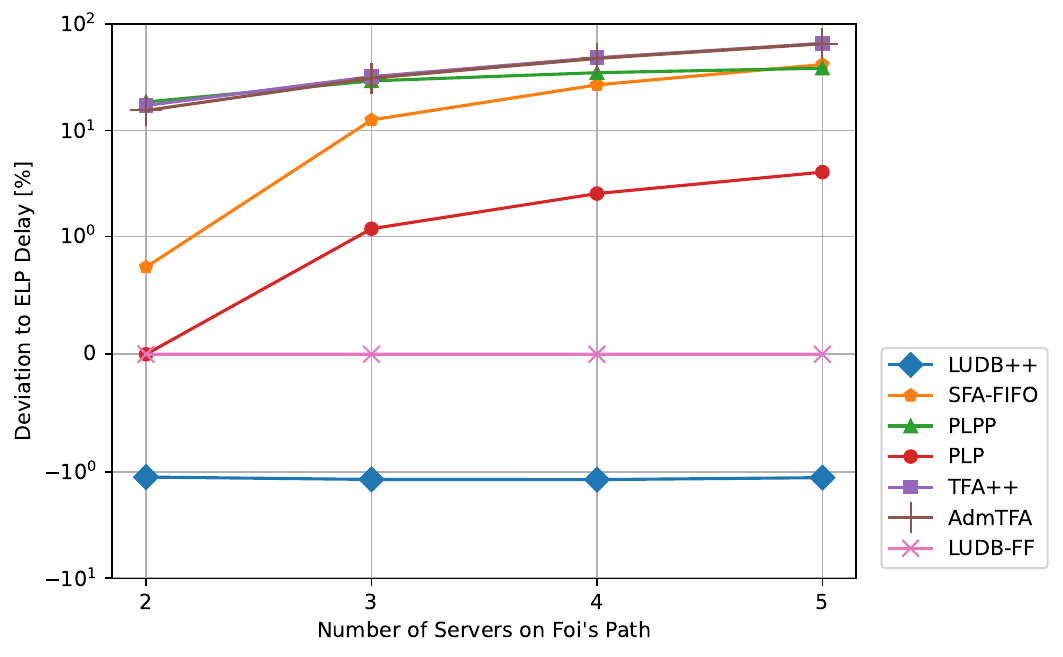} }}
    \\
    \subfloat[$u=75\%, R'=R$]{{\includegraphics[width=6cm]{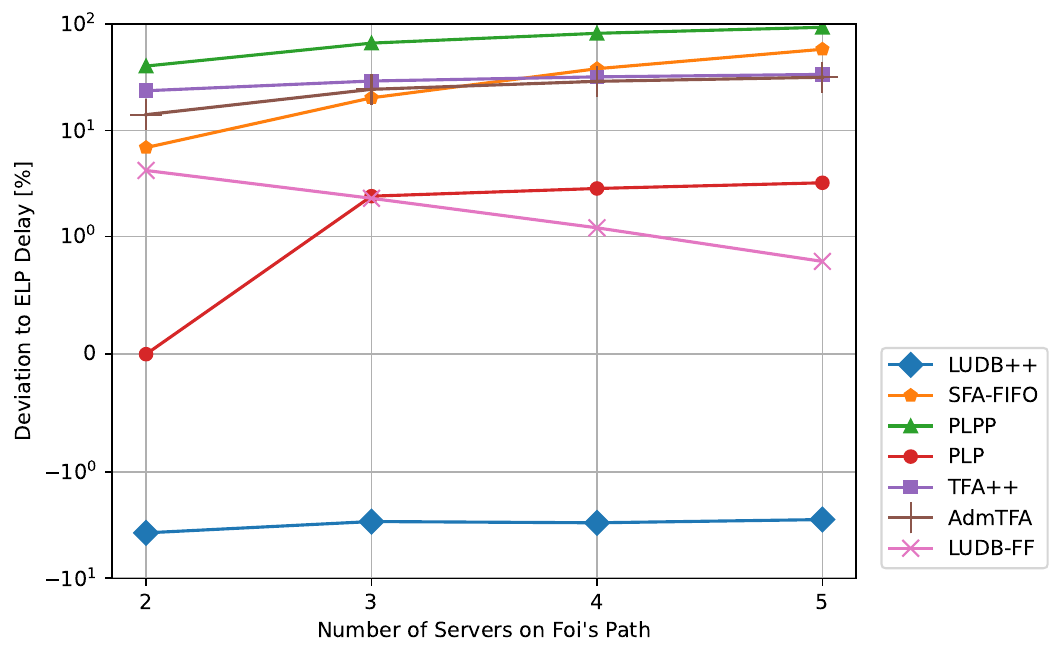} } }
    \subfloat[$u=75\%, R'=2R$]{{\includegraphics[width=6cm]{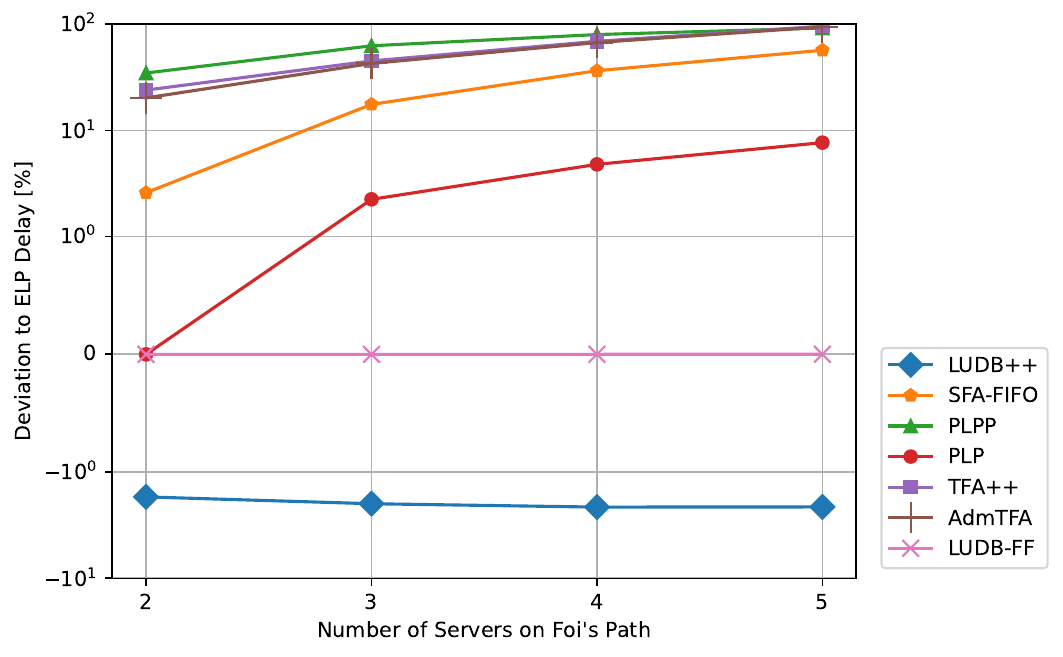} }}
    \subfloat[$u=75\%, R'=3R$]{{\includegraphics[width=6cm]{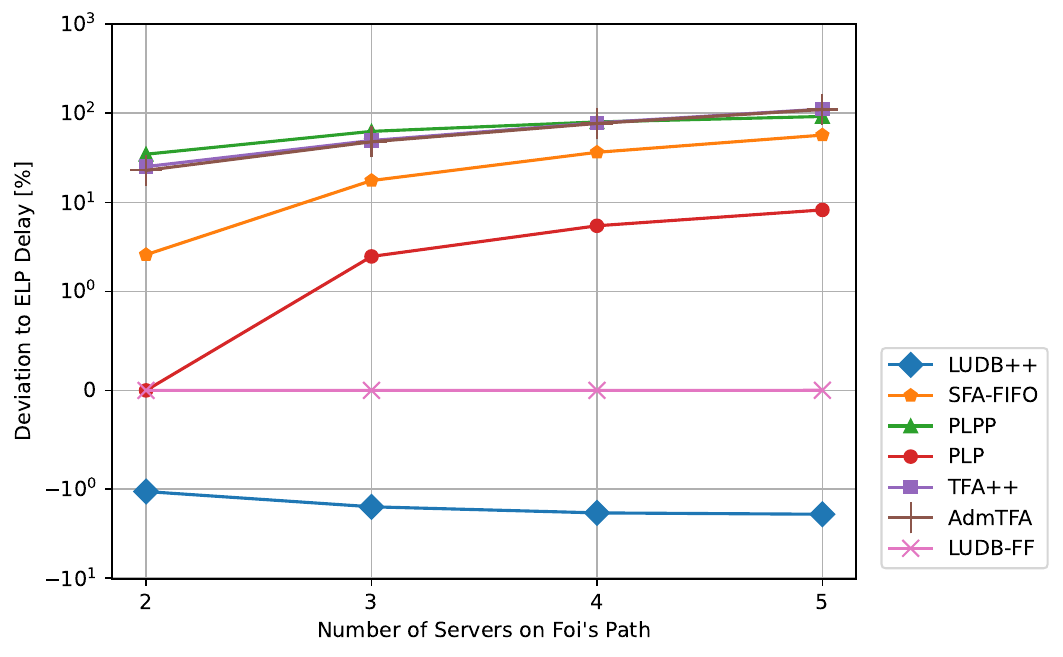} }}
    \\
    \subfloat[$u=100\%, R'=R$]{{\includegraphics[width=6cm]{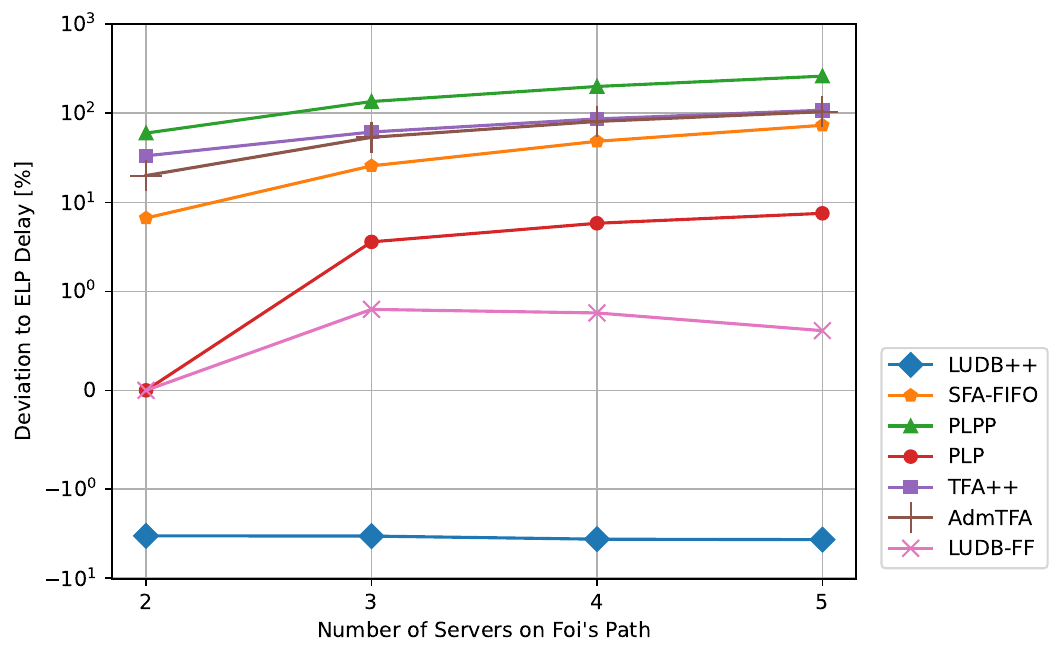} } }
    \subfloat[$u=100\%, R'=2R$]{{\includegraphics[width=6cm]{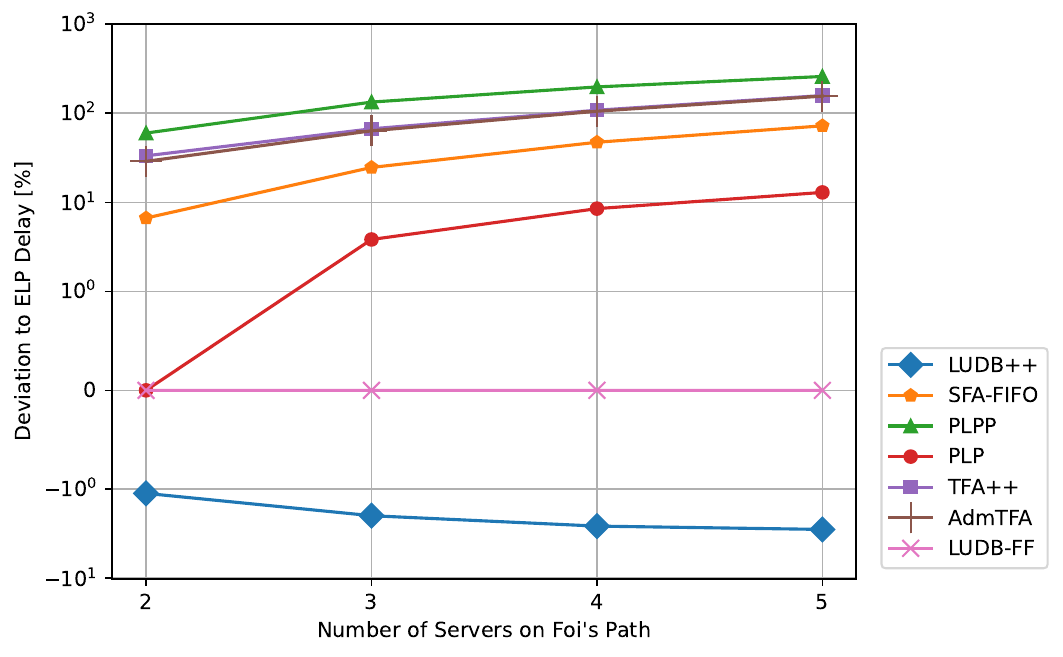} }}
    \subfloat[$u=100\%, R'=3R$]{{\includegraphics[width=6cm]{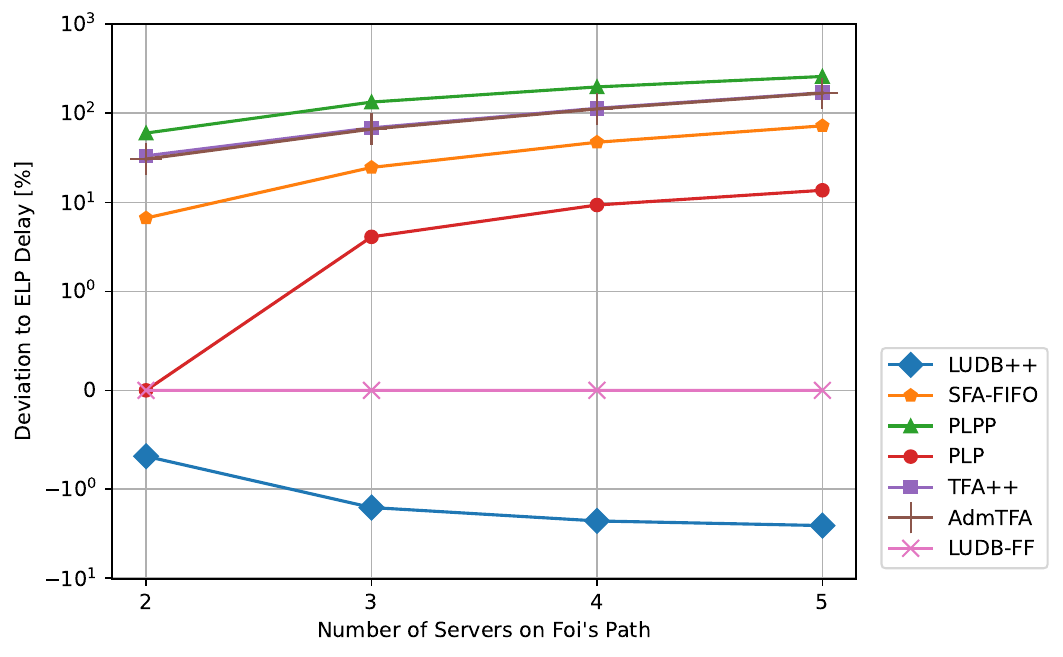} }}
    \caption{Relative delays to ELP for different shaper configurations, different utilizations and different $R'$ to $R$ ratios for the specified tree topology.}
    \label{fig:eval_tree}
\end{figure*}
LUDB++ beats ELP for all of the 36 settings. 
The minimal, maximal and average deviation of LUDB++ to ELP is $-3.74\%$, $-0.67\%$ and $-2.13\%$ respectively.
Here, ELP shows an influence on $R'$ (not shown in the figures) for this topology. 
%In general, ELP is link shaping aware and input shaping non-aware. 
%The same goes for PLP and PLPP. PLPP shows no influence on $R'$, but PLP does (not shown in the figures). For ELP and PLP an increase of $R'$ increases the delay bound in almost all scenarios (not shown in the figures).
For a fixed utilization, a decrease of $R'$ still tends to increase the absolute deviation of LUDB++ to ELP indicating that LUDB++ is more sensitive to a change of $R'$ than ELP.  
%In general, a smaller $R'$ has a larger impact on the shaped input arrival curve, decreasing the LUDB++ delay bounds and although ELP here also reacts to this change,  LUDB++ reacts more strongly, thus enlarging the gap to ELP. 
	%\item For a fixed $R'$ to $R$ ratio, decreasing the utilization tends to increases the delay gap between LUDB++ and ELP (not all cases but most of them). NO this is not true upon inspecting the data in detail.
%LUDB++ delay bounds are less than LUDB-FF as expected as the latter does not take shaping into account.
Interestingly, LUDB-FF delivers the same bound as ELP in 6 of the 9 settings for this topology. 
For $R'=R$, LUDB-FF delivers worse bounds (except the 2 server scenario in (g)) and the gap gets smaller by increasing the number of servers.
PLP delivers worse bounds compared to LUDB-FF except for $R'=R$ paired with lower than $100\%$ utilization and low number of servers (2 to 3). 
The gap of PLP to ELP increases with increasing number of servers for this topology.
SFA-FIFO delivers worse bounds compared to PLP and in some scenarios is beaten by TFA++ and AdmTFA, especially for  $R'=R$, mid to low utilization and from 3 servers upwards. 
For utilization of $50\%$, SFA-FIFO is beaten by PLPP for larger networks (5 servers) although the latter performs the worst compared to all other methods for the other scenarios.

\subsubsection{Runtimes}
\begin{figure*}[t!]
    \subfloat[One hop persistent tandems (ELP non-nan adjusted)]{{\includegraphics[width=6cm]{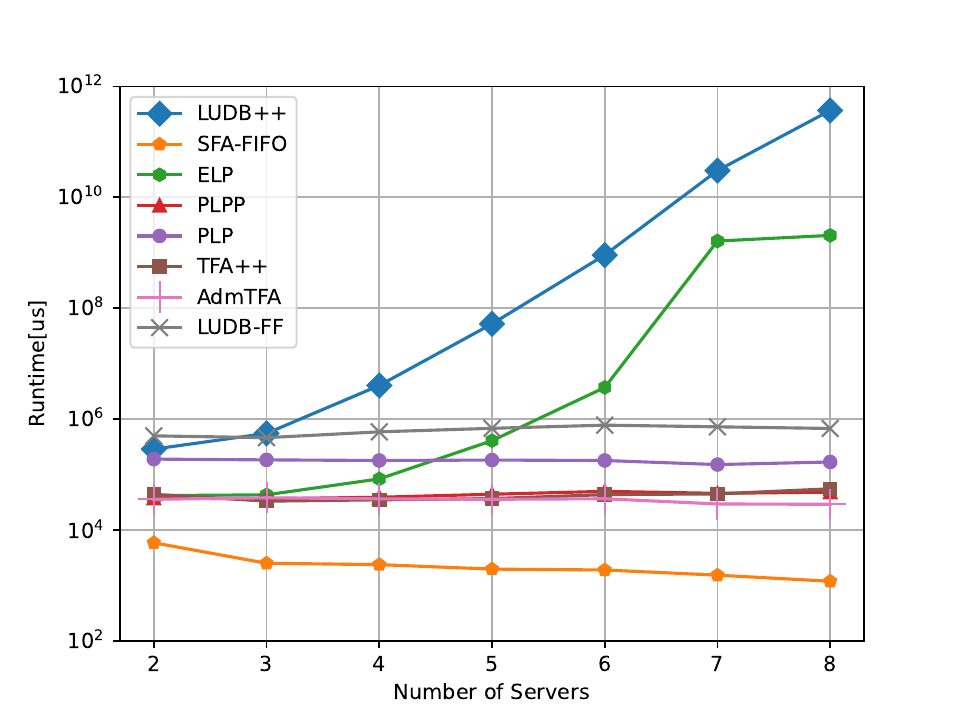} } }
    \subfloat[Sinktree tandems]{{\includegraphics[width=6cm]{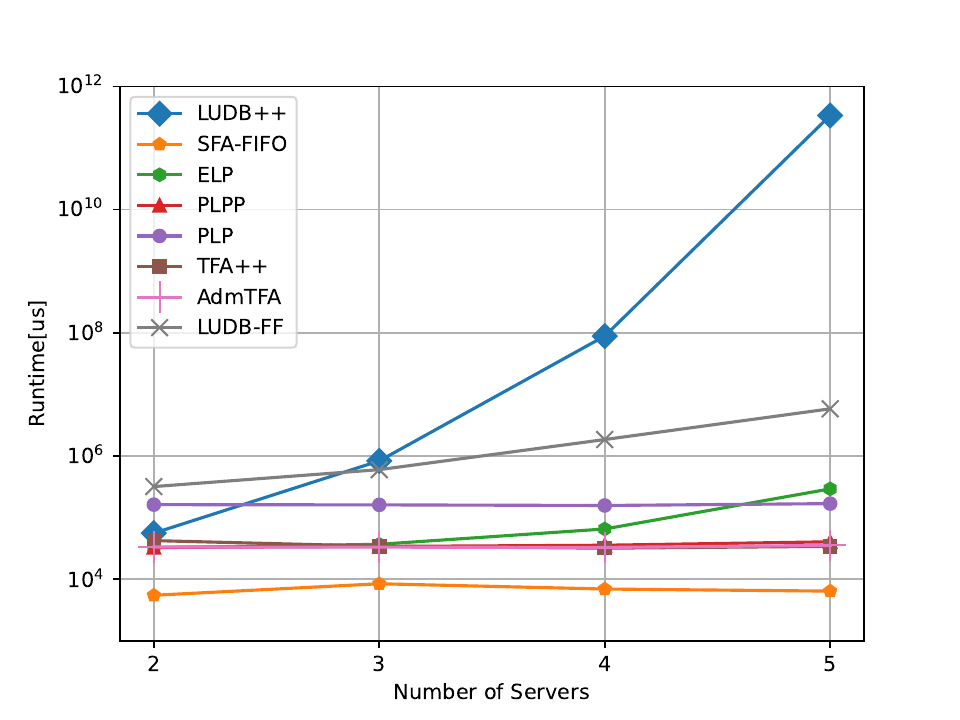} }}
     \subfloat[Tree]{{\includegraphics[width=6cm]{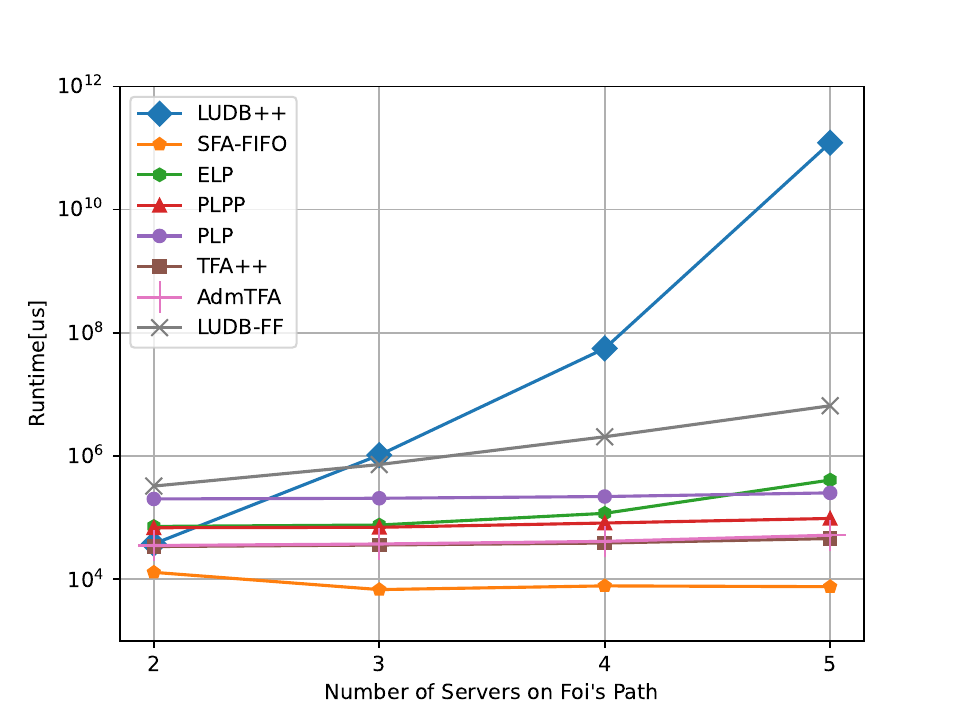} }}
    \caption{Aggregate runtimes per topology.}
    \label{fig:eval_runtimes_all}
\end{figure*}
In the following, we discuss the aggregate runtimes which means that for a fixed topology and server count up to 9 runs of a specific analysis such as LUDB++ is summed up.
For at least 3 servers, LUDB++ takes the longest compared to all other analyses.
Its runtime ranges from 0.29 s for 2 servers to 4.26 d for 8 servers regarding the one-hop persistent topology. 
Its runtimes for sinktree tandem (tree) topology range from 0.05 s (0.03 s)  for 2 servers to 3.91 d (1.41 d) for 5 servers.
Next is LUDB-FF with a runtime ranging from 0.5 s for 2 servers to 0.68 s for 8 servers regarding the one-hop persistent topology. % achieving an almost constant runtime for the one-hop persistent topology. 
The runtimes for sinktree tandem (tree) topology range from 0.32 s (0.32 s) for 2 servers to 5.83 s (6.54 s) for 5 servers.
That LUDB++'s runtime is larger than LUDB-FF is due to the considerably more complex network model which incorporates shaping assumptions and thus achieves better delay bounds.
ELP has a runtime ranging from 0.04 s for 2 servers to 34.22 min for 8 servers. 
Thus, for at least 5 servers upwards ELP takes longer than LUDB-FF for the one-hop persistent topology. 
Runtimes of ELP for sinktree tandem (tree) topology range from 0.03 s (0.07 s) for 2 servers to 0.29 s (0.41 s) for 5 servers.
The other methods are all below 1s with PLP being the most expensive one since it internally calls TFA++ and SFA and adds respective delay constraints into the optimization formulation before solving. 
PLPP's optimization problem is as PLP just without the TFA++ and SFA constraints and is thus faster but comes with a higher inaccuracy. 
SFA-FIFO is the fastest method among all for which no optimization is required. %, and can be described as a simple hop-by-hop analysis that does not take shaping into account at all. 

\section{Conclusion}\label{sec:conclusion}
This paper discusses formal delay bound guarantees in feedforward FIFO networks with shapers that limit the maximum packet size and rate of flows.
Our work extends LUDB, which is FIFO-aware but does not take shaping into account, is called LUDB++ and takes both aspects into consideration.
Hence, LUDB++ delivers better delay bounds than LUDB except in edge cases where both methods yield the same result.
ELP is another existing method that also considers both aspects in its analysis and encodes it into one exponential LP.
Our evaluation shows that LUDB++ delivers better delay bounds compared to ELP for the majority of the cases by up to 9.13$\%$.
However, these delay bounds come with a computational price. 
Although LUDB++ splits up the feedforward network into nested (sub-)tandems, solving these tandems can quickly get computationally expensive depending on the topology and number of crossflows.
This is especially the case for sinktree tandems where our evaluation shows that for 5 servers and a total of 9 configurations the total runtime already grows to 3.91 d.
Future work will focus on improving the scalability of the approach by e.g., approximation techniques.
One such technique could be to make an educated guess on the decomposition(s) during the analysis -- from the 4 possible ones per leftover operation, only consider $1 \leq k < 4$ ones.
AI could also be combined with this as it has been successfully combined with NC in the past in other contexts such as flow prolongation \cite{2023-GSB} for which a choice within the analysis is AI-driven instead of brute-forcing all alternatives. 
Moreover, LUDB++ computes for a vector of variables that each stem from an application of the FIFO leftover service theorem a setting that strictly minimizes the delay or backlog bound.
While preserving non-negativity, the variables can be set arbitrarily, so approximated solutions are valid as well and have been applied in the past to network models without shaping assumptions \cite{scheffler22rtns} \cite{scheffler22itc}.  
It might also be worthwhile to apply parallelization on the decompositions on the nested (sub-)tandem level within LUDB++ which is one analysis-level deeper than parallelization of the splitting the network into tandems process discussed in \cite{2018-SFB-1}.

\section*{Acknowledgments}
The author thanks Anne Bouillard for her help with the derivation of Equation \eqref{eq:plcs_n_1}.

{\appendix

We now present the proofs of the Theorems of Section \ref{sec:ludb_plus_plus}.

\begin{proof}[Proof of Lemma \ref{lemma:rl_in_mslc}]
$\beta=\beta_{R,T}=\delta_T \otimes ((\delta_0 \otimes I_0 \otimes \gamma_{0,R})\wedge \delta_0) \in$ \curveclass.
\end{proof}

\begin{lemma}[Convolution of \curveclass~with $n=1$  ]\label{lemma:mslc_n_1_conv}
	Let $\beta_1, \beta_2 \in$ \curveclass~with $\beta_1 =\delta_{D_1} \otimes ((\delta_{\tau_1} \otimes I_{\sigma_1} \otimes \gamma_{0,\rho_1})\wedge \delta_0)$ and $\beta_2 =\delta_{D_2} \otimes ((\delta_{\tau_2} \otimes I_{\sigma_2} \otimes \gamma_{0,\rho_2})\wedge \delta_0)$.
	Then $\beta_1 \otimes \beta_2 \in$ \curveclass.  
\end{lemma}
\begin{proof}
Due to commutativity of $\otimes$, we can write $\beta_1 \otimes \beta_2 = \delta_{D_1+D_2} \otimes ((\delta_{\tau_1} \otimes I_{\sigma_1} \otimes \gamma_{0,\rho_1})\wedge \delta_0) \otimes ((\delta_{\tau_2} \otimes I_{\sigma_2} \otimes \gamma_{0,\rho_2})\wedge \delta_0)$. 
Since $\beta_{D_1+D_2}$ is just shifting the function by $D_1+D_2$ units to the right, we focus on the expression $[(\delta_{\tau_1} \otimes I_{\sigma_1} \otimes \gamma_{0,\rho_1})\wedge \delta_{0}] \otimes [(\delta_{\tau_2} \otimes I_{\sigma_2} \otimes \gamma_{0,\rho_2})\wedge \delta_0)]=:f \otimes g (t) = \inf_{\substack{r,s\geq 0 \\ r+s=t}} \{f(r)+g(s)\} $. \\
Case 1 ($t=0$): $f(0)\wedge g(0) = 0 = f(0)+g(0).$\\\\
Case 2 ($0 < t \leq \max\{ \tau_1, \tau_2\}$): 
	$\inf_{s=0}\{f(t)+g(0)\} \wedge \inf_{t=s}\{f(0)+g(t)\}\wedge \inf_{0<s<t}\{f(t-s)+g(s)\} = f(t) \wedge g(t) \wedge \inf_{0<s<t}\{f(t-s)+g(s)\}$. \\
Now, we evaluate  $\inf\limits_{0<s< t \leq \max\{ \tau_1, \tau_2\}} \{f(t-s)+g(s)\}$.
First of all, note that $g(s) \geq \sigma_2$ and $f(t-s)\geq \sigma_1$. 
We  have to show that $f(t-s)+g(s) \geq f(t) \wedge g(t)$ for $t \geq \max\{\tau_1, \tau_2\}$ and $t > 0$. \\
$g(s)= \begin{cases}
		\sigma_2  & \text{if } 0 < s \leq \tau_2 \\
        \sigma_2+\rho_2(s-\tau_2) & \text{if } s>\tau_2, s < t \text{ which implies } \\
        						  &	 \max\{\tau_1, \tau_2\}=\tau_1
    \end{cases}$ 
 \\
$f(t-s)= \begin{cases}
		\sigma_1  & \text{if } 0 < t-s \leq \tau_1 \\
        \sigma_1+\rho_1(t-s-\tau_1) & \text{if } t-s>\tau_1,  \\
        							&t-s < \max\{\tau_1, \tau_2\} \text{ which }\\
        							& \text{ implies } \max\{\tau_1, \tau_2\}=\tau_2
    \end{cases}$ 
    \\
Subcase 2.1 ($\max\{\tau_1, \tau_2\}=\tau_1$): $f(t-s)=\sigma_1, g(s)=\sigma_2$ and $f(t)\wedge g(t) = \sigma_1 \wedge (\sigma_2 + \rho_2 (t-\tau_2)), f(t-s)+g(s)=\sigma_1 + \sigma_2 \geq \sigma_1 \wedge (\sigma_2 + \rho_2 (t-\tau_2))$. \\
Subcase 2.2 ($\max\{\tau_1, \tau_2\}=\tau_2$): $f(t-s)=\sigma_1, g(s)=\sigma_2$ and $f(t)\wedge g(t) =  (\sigma_1 + \rho_1 (t-\tau_1)) \wedge \sigma_2 , f(t-s)+g(s)=\sigma_1 + \sigma_2 \geq (\sigma_1 + \rho_1 (t-\tau_1)) \wedge \sigma_2$. \\
Hence, for this case $f \otimes g (t) = f(t)\wedge g(t)$. \\\\
Case 3 ($t > \max\{ \tau_1, \tau_2\} >0$): As before, we have
$ f \otimes g (t) = f(t) \wedge g(t) \wedge \inf_{0<s<t}\{f(t-s)+g(s)\}$. 
Since $t > \max\{ \tau_1, \tau_2\}$, $f(t) = \sigma_1 + \rho_1 (t-\tau_1)$ and $g(t) = \sigma_2 + \rho_2 (t- \tau_2)$. 
Moreover, note that $f(t-s) \geq \sigma_1$ and $g(s) \geq \sigma_2$ for $0 < s < t$. \\
Subcase 3.1 ($t \leq \tau_1 + \tau_2$): $\inf_{0<s<t}\{f(t-s)+g(s)\}=\sigma_1 + \sigma_2$ \\
Subcase 3.2 ($t > \tau_1 + \tau_2$): $\inf_{0<s<t}\{f(t-s)+g(s)\}=(\sigma_1 + (t-\tau_2-\tau_1)\rho_1 + \sigma_2)\wedge(\sigma_2+(t-\tau_1-\tau_2)\rho_2 +\sigma_1)$. \\
Hence, for this case we get 
$f \otimes g (t) = \begin{cases}
	f(t) \wedge g(t) \wedge (\sigma_1 + \sigma_2)  & \text{if } 0 < \max\{ \tau_1, \tau_2 \} < t \\
													& \text{and }t \leq \tau_1 + \tau_2 \\
	f(t) \wedge g(t) \wedge [f(t-\tau_2)+\sigma_2] & \text{if } 0 <  \tau_1 + \tau_2 < t  \\
	\wedge [g(t-\tau_1)+\sigma_1]  
\end{cases}$
\\\\
Finally, note that we can express $f \otimes g$ as minimum of 4 curves each of type $((\delta_{\tau_k} \otimes I_{\sigma_k} \otimes \gamma_{0,\rho_k})\wedge \delta_0)$ as follows.
$f \otimes g (t)$ \\
$ =  ((\delta_{\tau_1} \otimes I_{\sigma_1} \otimes \gamma_{0,\rho_1})\wedge \delta_0) \otimes ((\delta_{\tau_2} \otimes I_{\sigma_2} \otimes \gamma_{0,\rho_2})\wedge \delta_0) $\\
$=  ((\delta_{\tau_1} \otimes I_{\sigma_1} \otimes \gamma_{0,\rho_1})\wedge \delta_0) \wedge ((\delta_{\tau_2} \otimes I_{\sigma_2} \otimes \gamma_{0,\rho_2})\wedge \delta_0) $\\
$~~~~\wedge ((\delta_{\tau_1+\tau_2} \otimes I_{\sigma_1+\sigma_2} \otimes \gamma_{0,\rho_1}) \wedge \delta_0) $\\
$~~~~ \wedge  ((\delta_{\tau_1+\tau_2} \otimes I_{\sigma_1+\sigma_2} \otimes \gamma_{0,\rho_2})\wedge \delta_0) $\\
$= f (t) \wedge g (t) \wedge (f (t-\tau_2) + \sigma_2) \wedge (g (t-\tau_1) + \sigma_1) $ \\
The following figure illustrates the convolution of $f$ and $g$.
\begin{figure}[H]  
\begin{centering}
  \includegraphics[width=0.5\textwidth]{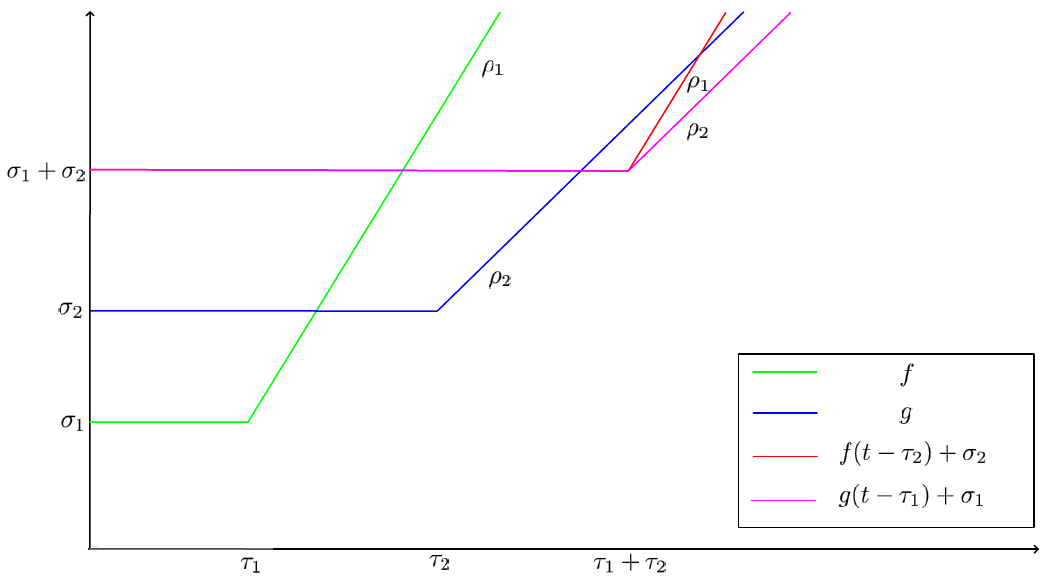}  
  \caption{Convolution of $f$ and $g$.}
\end{centering}
\end{figure} 
Note that the result also holds if $\rho_1$ or $\rho_2$ is infinite.
Since $\beta_1 \otimes \beta_2 = \delta_{D_1+D_2} \otimes (f \otimes g)$, we can now conclude that $\beta_1 \otimes \beta_2 \in$ \curveclass~and the proof is complete.
\end{proof}

\begin{proof}[Proof of Theorem \ref{theorem:mslc_closed_under_conv}]
\leavevmode\\
$ \beta_1 \otimes \beta_2 $\\
$ = [\delta_{D_1}\otimes (\bigwedge\limits_{i=1}^{n_1}(\delta_{\tau_i} \otimes I_{\sigma_i} \otimes \gamma_{0, \rho_i})\wedge \delta_0)] \otimes $ \\
$ ~~~~[\delta_{D_2}\otimes (\bigwedge\limits_{i=1}^{n_2}(\delta_{\tilde{\tau_i}} \otimes I_{\tilde{\sigma_i}} \otimes \gamma_{0, \tilde{\rho_i}})\wedge \delta_0)] $\\
$\stackrel{(*)}{=} \delta_{D_1 + D_2} \otimes 	(\bigwedge\limits_{i=1}^{n_1}(\delta_{\tau_i} \otimes  I_{\sigma_i} \otimes  \gamma_{0, \rho_i})\wedge \delta_0) \otimes $ \\
$ ~~~~(\bigwedge\limits_{i=1}^{n_2}(\delta_{\tilde{\tau_i}} \otimes I_{\tilde{\sigma_i}} \otimes \gamma_{0, \tilde{\rho_i}})\wedge \delta_0)$ \\
$\stackrel{(**)}{=}  \delta_{D_1 + D_2} \otimes (\bigwedge\limits_{i=1}^{n_1} \bigwedge\limits_{j=1}^{n_2} ((\delta_{\tau_i} \otimes I_{\sigma_i} \otimes \gamma_{0, \rho_i})\wedge \delta_0)  \otimes $\\ 
$~~~~((\delta_{\tilde{\tau_j}} \otimes I_{\tilde{\sigma_j}} \otimes \gamma_{0, \tilde{\rho_j}})\wedge \delta_0)) $\\
$\stackrel{(***)}{=} \delta_{D_1 + D_2} \otimes (\bigwedge\limits_{i=1}^{n_1} \bigwedge\limits_{j=1}^{n_2} (((\delta_{\tau_i} \otimes I_{\sigma_i} \otimes \gamma_{0, \rho_i})\wedge \delta_0)  \wedge$\\
$~~~~ ((\delta_{\tilde{\tau_j}} \otimes I_{\tilde{\sigma_j}} \otimes \gamma_{0, \tilde{\rho_j}})\wedge \delta_0) \wedge  ((\delta_{\tau_i + \tilde{\tau_j}} \otimes I_{\sigma_i + \tilde{\sigma_j}} \otimes \gamma_{0, \rho_i})\wedge \delta_0) $\\
$~~~~\wedge  ((\delta_{\tau_i + \tilde{\tau_j}} \otimes I_{\sigma_i + \tilde{\sigma_j}} \otimes \gamma_{0, \tilde{\rho_j}})\wedge \delta_0)))$. \\
$(*)$ holds due to commutativity of $\otimes$, $(**)$ holds due to $(f\wedge g)\otimes h=(f\otimes h)\wedge(g \otimes h)$ (see \cite{le2001network}) and  $(***)$ holds due to Lemma \ref{lemma:mslc_n_1_conv}.
Hence, we have shown that $\beta_1 \otimes \beta_2 \in$ \curveclass~and the proof is complete.
\end{proof}

\begin{lemma}[Delay Bound of an Arrival Curve and Shifted Service Curve]\label{lemma:delay_bound_ac_shifted_sc}
	Let $\alpha$ be an arrival curve and $\beta$ a service curve and $D \geq 0$. 
	Then $hdev(\alpha, \delta_D \otimes \beta)=D+hdev(\alpha,\beta)$.
\end{lemma}
\begin{proof} \leavevmode\\
	$hdev(\alpha, \delta_D \otimes \beta)$ \\
	$= \sup\limits_{t\geq 0} [inf\{d\geq 0 | \alpha(t-d) \leq \delta_D \otimes \beta(t) \}]$ \\
	$= \sup\limits_{t\geq 0} [inf\{d\geq 0 | \alpha(t-d) \leq \beta(t-D) \}] $\\
	$\stackrel{(*)}{=} \sup\limits_{t\geq D} [inf\{d\geq 0 | \alpha(t-d) \leq \beta(t-D) \}] $\\
	$\stackrel{(**)}{=} \sup\limits_{t\geq D} [inf\{d\geq 0 | \alpha(D+t'-d) \leq \beta(t') \}] $\\
	$= \sup\limits_{t\geq 0} [inf\{d\geq 0 | \alpha(D+t-d) \leq \beta(t) \}] $\\
	$= \sup\limits_{t\geq 0} [inf\{d\geq 0 | \alpha(t-(d-D)) \leq \beta(t) \}] $\\
	$= \sup\limits_{t\geq 0} [D+inf\{d\geq 0 | \alpha(t-d) \leq \beta(t) \}] $\\
	$= D+ \sup\limits_{t\geq 0} [inf\{d\geq 0 | \alpha(t-d) \leq \beta(t) \}] $\\
	$= D+ hdev(\alpha,\beta)$ \\
	$(*)$ holds since $\beta(t-D)=0 \text{ for } t\leq D$, $(**)$ holds with $t=D+t'$.
\end{proof}

\begin{lemma}[Delay Bound of an Arrival Curve and Minima of Several Service Curves]\label{lemma:delay_bound_min_service_curve_eq_max_delay_bounds_individually}
Let $\alpha$ be an arrival curve and $\beta$ a service curve of the form $\beta=\bigwedge\limits_{i=1}^n \beta_i$. 
Then $hdev(\alpha,\bigwedge\limits_{i=1}^n \beta_i)=\bigvee\limits_{i=1}^n hdev(\alpha, \beta_i)$.
\end{lemma}
\begin{proof} \leavevmode\\
$ hdev(\alpha,\bigwedge\limits_{i=1}^n \beta_i)$\\
$= \sup\limits_{t\geq 0} [inf\{d\geq 0 | \alpha(t-d) \leq \bigwedge\limits_{i=1}^n \beta_i(t) \}] $\\
$\stackrel{(*)}{=} \sup\limits_{t\geq 0} [\bigvee\limits_{i=1}^n inf\{d\geq 0 | \alpha(t-d) \leq \beta_i(t) \}]  $\\
$= \bigvee\limits_{i=1}^n \sup\limits_{t\geq 0} [inf\{d\geq 0 | \alpha(t-d) \leq \beta_i(t) \}] $\\
$= \bigvee\limits_{i=1}^n hdev(\alpha, \beta_i) $\\
$(*)$ due to $\alpha \in \cal{F}$, i.e., a wide sense increasing function.
\end{proof}

\begin{lemma}[Delay Bound of Shaped Arrival Curve and \curveclass~Service Curve with $n=1$ and $D=0$]\label{lemma:delay_bound_shaped_ac_simple_plcs_zero_offset}
	Let $\alpha=\gamma_{b,r} \wedge \gamma_{L,R'}$ and $\beta \in$ \curveclass~with $n=1$ and $D=0$, i.e., $\beta = (\delta_{\tau_1}\otimes I_{\sigma_1}\otimes \gamma_{0,\rho_1})\wedge \delta_0$, with $L \leq b, r < R', R' \geq \rho_1$ (if $\rho_1 < +\infty$) and $r \leq \rho_1$. \\ 
	Then, $hdev(\alpha, \beta) =$ \\
	$ \begin{cases}
		[\tau_1 \cdot 1_{ \{ \sigma_1 \geq \frac{b-L}{R'-r}r+b \}} - \frac{[\sigma_1 -b]^+}{r},  & \text{if } \rho_1 <+\infty\\
		 \tau_1 \cdot 1_{\{\frac{b-L}{R'-r}r+b\geq \sigma_1 \}} +  \frac{[\frac{b-L}{R'-r}r+b-\sigma_1]^+}{\rho_1} \\
		 - \frac{b-L}{R'-r}]^+   \\
        [\tau_1 \cdot 1_{ \{ \sigma_1 \geq \frac{b-L}{R'-r}r+b \}} - \frac{[\sigma_1 -b]^+}{r},  & \text{if }  \rho_1 = +\infty\\
         \tau_1 \cdot 1_{\{\frac{b-L}{R'-r}r+b\geq \sigma_1 \}} - \frac{[\sigma_1-L]^+}{R'}]^+   
    \end{cases}$
\end{lemma}
\begin{proof} 
We first start with the assumption $\rho_1 < +\infty$ and note that $\frac{b-L}{R'-r}r+b$ is the $y$-value of the inflection point of $\alpha$.

Case 1 ($\sigma_1 \leq \frac{b-L}{R'-r}r+b$): 
$hdev$ can only be between the inflection value of $\alpha$ and the respective value of $\beta$ since $R' \geq \rho_1$.
This can be described as $\tau_1 + \frac{\frac{b-L}{R'-r}r + b-\sigma_1}{\rho_1}-\frac{b-L}{R'-r}$.
\begin{figure}[H]  
\begin{centering}
  \includegraphics[width=0.5\textwidth]{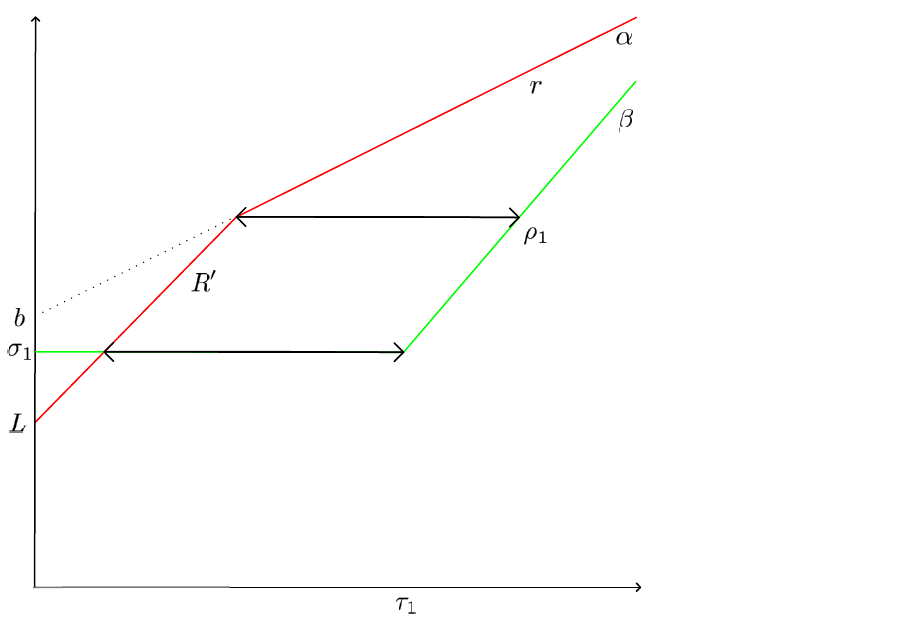}  
  \caption{Case 1 of the delay bound proof.}
\end{centering}
\end{figure} 

Case 2 ($\sigma_1 > \frac{b-L}{R'-r}r+b$):
$hdev$ can only be between the inflection point of $\beta$ and the respective value of $\alpha$ as $r \leq \rho_1$.
Moreover, we have to guard this case also for negativity, as theoretically, $\alpha$ may not have any intersection with $\beta$.
This can be described as $[\tau_1 - \frac{\sigma_1 -b}{r}]^+$. 
\\
%TODO if enough space, include the fig
%\begin{figure}[H]  
%\begin{centering}
%  \includegraphics[width=0.5\textwidth]{figures/hdev_shaped_ac_plcs_n_indices_1_case_2_a}  
%  \caption{Case 2.}
%\end{centering}
%\end{figure} 
It is then easily verifiable, that both cases can be summarized as $[\tau_1 \cdot 1_{ \{ \sigma_1 \geq \frac{b-L}{R'-r}r+b \}} - \frac{[\sigma_1 -b]^+}{r},  \tau_1 \cdot 1_{\{\frac{b-L}{R'-r}r+b\geq \sigma_1 \}} +  \frac{[\frac{b-L}{R'-r}r+b-\sigma_1]^+}{\rho_1} - \frac{b-L}{R'-r}]^+ $.
\\\\
Finally, we now assume that $\rho_1 = +\infty$.
\\
Case 1 ($\sigma_1 \leq \frac{b-L}{R'-r}r+b$): If $L > \sigma_1$, then $hdev = \tau_1$ since $R' < +\infty$. Otherwise, $hdev=[\tau_1-\frac{\sigma_1-L}{R'}]^+$.  \\
Case 2 ($\sigma_1 > \frac{b-L}{R'-r}r+b$): The delay bound is equal to $[\tau_1-\frac{\sigma_1-b}{r}]^+$. \\
It is then again easily verifiable, that both cases can be summarized as $[\tau_1 \cdot 1_{ \{ \sigma_1 \geq \frac{b-L}{R'-r}r+b \}} - \frac{[\sigma_1 -b]^+}{r}, \tau_1 \cdot 1_{\{\frac{b-L}{R'-r}r+b\geq \sigma_1 \}} - \frac{[\sigma_1-L]^+}{R'}]^+$.
\end{proof}

\begin{proof}[Proof of Theorem \ref{theorem:delay_bound_shaped_ac_plcs_sc}]
	We divide the proof as follows:
	\begin{itemize}
	\item $hdev(\alpha,\delta_D \otimes \beta)= D + hdev(\alpha, \beta)$ for any arrival curve $\alpha$ and service curve $\beta$ (Lemma \ref{lemma:delay_bound_ac_shifted_sc})
	\item $hdev(\alpha, \bigwedge\limits_{i=1}^n \beta_i)= \bigvee\limits_{i=1}^n hdev(\alpha,\beta_i)$ for any arrival curve $\alpha$ and service curve $\beta$ (Lemma \ref{lemma:delay_bound_min_service_curve_eq_max_delay_bounds_individually})
	\item  $hdev(\alpha,\beta_i) = $\\
		$\begin{cases}
		[\tau_i \cdot 1_{ \{ \sigma_i \geq \frac{b-L}{R'-r}r+b \}} - \frac{[\sigma_i -b]^+}{r}, & \text{if } \rho_i <+\infty \\
		 \tau_i \cdot 1_{\{\frac{b-L}{R'-r}r+b\geq \sigma_i \}} + \frac{[\frac{b-L}{R'-r}r+b-\sigma_i]^+}{\rho_i} \\
		  - \frac{b-L}{R'-r}]^+  \\
        [\tau_i \cdot 1_{ \{ \sigma_i \geq \frac{b-L}{R'-r}r+b \}} - \frac{[\sigma_i -b]^+}{r}, & \text{if }  \rho_i = +\infty\\
        \tau_i \cdot 1_{\{\frac{b-L}{R'-r}r+b\geq \sigma_i \}} - \frac{[\sigma_i-L]^+}{R'}]^+   
    \end{cases}$  (Lemma \ref{lemma:delay_bound_shaped_ac_simple_plcs_zero_offset})
\end{itemize}
Hence, all in all, we obtain the claim.
\end{proof}

\begin{lemma}[Backlog Bound Arrival Curve and \curveclass~Service Curve with $n=1$]\label{lemma:backlog_bound_simple_ac_plcs_n_1}
	Consider an arrival curve with $\alpha=\gamma_{b,r}$ and a service curve $\beta$ in \curveclass~with $n=1$, i.e., $\beta$ has the form $\delta_D \otimes ((\delta_{\tau_1} \otimes I_{\sigma_1}\otimes \gamma_{0,\rho_1}) \wedge \delta_0)$. 
	Then, $vdev(\alpha,\beta) =[b+Dr, b-\sigma_1+(D+\tau_1)r]$.
\end{lemma}
\begin{proof}
$vdev(\alpha,\beta)$ can be computed as \\
$\sup\limits_{u\geq 0} \{\alpha(u)-\beta(u)\}$\\
$\stackrel{(*)}{=} [\alpha(D)-0] \vee [\alpha(D+\tau_1)-\beta(D+\tau_1)] $ \\
$= [\alpha(D)] \vee [\alpha(D+\tau_1)-\sigma_1] $\\
$= [b+Dr] \vee [b-\sigma_1+(D+\tau_1)r] $\\
$= [b+Dr, b-\sigma_1 + (D+\tau_1)r]$\\
$(*)$ since $u=D$ or $u=D+\tau_1$.

\begin{figure}[H]  
\begin{centering}
  \includegraphics[width=0.5\textwidth]{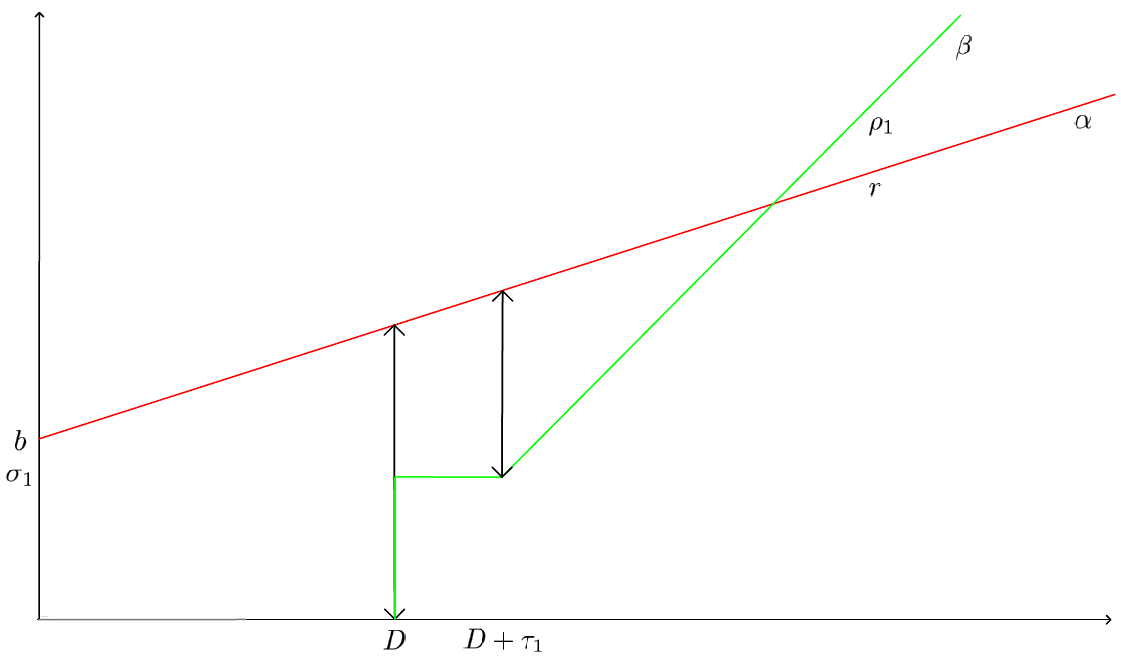}  
  \caption{Vertical deviation between $\alpha$ and $\beta$.}
\end{centering}
\end{figure}
\end{proof}

\begin{lemma}[Vertical Deviation of Arrival Curve and \curveclass~Service Curve equals Maximum of Vertical Deviations of Arrival Curve and  \curveclass~with $n=1$]\label{lemma:vdev_simple_ac_plcs_sc_equals_vdev_simple_ac_individual_plcs_parts}
	Let $\alpha=\gamma_{b,r}$ and $\beta \in$ \curveclass~of the form $\beta=\delta_D \otimes (\bigwedge\limits_{i=1}^n (\delta_{\tau_i}\otimes I_{\sigma_i}\otimes \gamma_{0,\rho_i})\wedge\delta_0)$ with $0 \leq r \leq \min\limits_{i=1,...,n} \rho_i$. 
	Then, $vdev(\alpha,\beta)=\bigvee\limits_{i=1}^n vdev(\alpha, \beta_i)$ with $\beta_i=\delta_D \otimes ((\delta_{\tau_i} \otimes I_{\sigma_i}\otimes \gamma_{0,\rho_i}) \wedge \delta_0)$.
\end{lemma}
\begin{proof} \leavevmode\\
$vdev(\alpha,\beta) $\\
$=\sup\limits_{u\geq 0}\{\alpha(u)-\beta(u)\} $\\
$= \sup\limits_{u\geq 0}\{\alpha(u)-\delta_D \otimes (\bigwedge\limits_{i=1}^n ((\delta_{\tau_i}\otimes I_{\sigma_i}\otimes \gamma_{0,\rho_i})\wedge\delta_0)(u))\} $\\
$\stackrel{\alpha=\gamma_{b,r}}{=} [b+Dr] \vee \sup\limits_{u\geq 0}\{\alpha(u+D)-(\bigwedge\limits_{i=1}^n ((\delta_{\tau_i}\otimes I_{\sigma_i}\otimes \gamma_{0,\rho_i})\wedge\delta_0)(u))\} $\\
$= [b+Dr] \vee \sup\limits_{u\geq 0}\{\bigvee\limits_{i=1}^n \alpha(u+D)-(((\delta_{\tau_i}\otimes I_{\sigma_i}\otimes \gamma_{0,\rho_i})\wedge\delta_0)(u))\} $\\
$\stackrel{\alpha=\gamma_{b,r}}{=} [b+Dr] \vee \sup\limits_{u\geq 0}\{\bigvee\limits_{i=1}^n Dr + \alpha(u)-(((\delta_{\tau_i}\otimes I_{\sigma_i}\otimes \gamma_{0,\rho_i})\wedge\delta_0)(u))\}$\\
$= [b+Dr] \vee \{Dr + \bigvee\limits_{i=1}^n  \{ \sup\limits_{u\geq 0}  \alpha(u)-(((\delta_{\tau_i}\otimes I_{\sigma_i}\otimes \gamma_{0,\rho_i})\wedge\delta_0)(u))\}\}$\\
$\stackrel{\text{Lemma } \ref{lemma:backlog_bound_simple_ac_plcs_n_1}}{=} [b+Dr] \vee \{ Dr + \bigvee\limits_{i=1}^n [b, b-\sigma_i + \tau_i r] \}$\\
$= [b+Dr] \vee \bigvee\limits_{i=1}^n [b+Dr, b-\sigma_i+(D+\tau_i)r] 	$\\
$=\bigvee\limits_{i=1}^n[b+Dr, b-\sigma_i+(D+\tau_i)r]$ \\
$= \bigvee\limits_{i=1}^n vdev(\alpha,\beta_i)$\\
\end{proof}

\begin{lemma}[Outputbound via Vertical Deviation for Arrival Curve and \curveclass~Service Curve]\label{lemma:outputbound_via_vertical_dev_simple_ac_plcs_sc}
	Let $\alpha=\gamma_{b,r}$ and $\beta \in$ \curveclass~ of the form $\beta=\delta_D \otimes (\bigwedge\limits_{i=1}^n (\delta_{\tau_i}\otimes I_{\sigma_i}\otimes \gamma_{0,\rho_i})\wedge\delta_0)$ with $0 \leq r \leq \min\limits_{i=1,...,n} \rho_i$. 
	Then, $\alpha'(t)=\alpha \oslash \beta(t)=rt+vdev(\alpha,\beta)$ for $t \geq 0$.
\end{lemma}
\begin{proof}
First of all, note that $\alpha \oslash \beta (0) = \sup\limits_{u\geq0}\{\alpha(u)-\beta(u)\}=vdev(\alpha,\beta)$. 
Now, let $t > 0$. 
Then, \\
$\alpha \oslash \beta(t) =  \sup\limits_{u\geq0}\{\alpha(t+u)-\beta(u)\} $\\
$\stackrel{\alpha=\gamma_{b,r}, t>0}{=} \sup\limits_{u\geq0}\{b+r(t+u)-\beta(u)\} $\\
$\stackrel{\beta(0)=0}{=} rt + \sup\limits_{u>0} \{\alpha(u)-\beta(u)\} \vee [b-\beta(0)] $\\
$\stackrel{\beta(0)=0}{=} rt +  \sup\limits_{u>0} \{\alpha(u)-\beta(u)\} \vee [b] $\\
$\stackrel{\alpha(0)=\beta(0)=0}{=} rt + vdev(\alpha,\beta)\vee[b] $\\
$\stackrel{(*)}{=} rt + vdev(\alpha,\beta). $ \\
$(*)$ since due to Lemma \ref{lemma:backlog_bound_simple_ac_plcs_n_1},  $vdev(\alpha,\beta)\geq b$ holds.
\end{proof}

\begin{proof}[Proof of Theorem \ref{theorem:backlog_bound_ouput_ac_simple_ac_plcs_sc}]
	We divide the proof into several steps
	\begin{itemize}
		\item $vdev(\alpha,\beta)=\bigvee\limits_{i=1}^n vdev(\alpha,\beta_i)$ for $\alpha, \beta$ as above (see Lemma \ref{lemma:vdev_simple_ac_plcs_sc_equals_vdev_simple_ac_individual_plcs_parts})
		\item $vdev(\alpha,\beta_i)=[b+Dr]\vee[b-\sigma_i+(D+\tau_i)r]$ (see Lemma \ref{lemma:backlog_bound_simple_ac_plcs_n_1})
		\item $\alpha\oslash \beta(t)=rt+vdev(\alpha,\beta)$ for $t \geq 0$ and $\alpha, \beta$ as above (see Lemma \ref{lemma:outputbound_via_vertical_dev_simple_ac_plcs_sc})
	\end{itemize}
\end{proof}

\begin{lemma}[Leftover of Minima of Service Curves] \label{lemma:left_over_minima_of_service_curves_equals_minima_of_left_over_service_curves}
%$LO(\alpha, \bigwedge\limits_{i=1}^n \beta_i, \theta) = \bigwedge\limits_{i=1}^n LO(\alpha,\beta_i, \theta)$ for any $\theta \geq 0$.
For any arrival curve $\alpha$ and service curve $\beta$, $(\bigwedge\limits_{i=1}^n \beta_i) \ominus_{\theta} \alpha = \bigwedge\limits_{i=1}^n (\beta_i \ominus_{\theta} \alpha)$ holds.
\end{lemma}
\begin{proof}
	First assume $t>\theta$ since otherwise the claim is trivial. 
	We write $\bigwedge\limits_{i=1}^n \beta_i(t) =:\beta_j(t)$. \\
	If $\beta_j(t) \leq \alpha(t-\theta)$, then $[(\bigwedge\limits_{i=1}^n \beta_i(t))-\alpha(t-\theta)]^+=[\beta_j(t)-\alpha(t-\theta)]^+=0$. Also $\bigwedge\limits_{i=1}^n[\beta_i(t)-\alpha(t-\theta)]^+=0$ since $[\beta_j(t)-\alpha(t-\theta)]^+=0, [\beta_i-\alpha(t-\theta)]^+\geq 0$. \\
	If $\beta_j(t) > \alpha(t-\theta)$,  then $[(\bigwedge\limits_{i=1}^n \beta_i(t))-\alpha(t-\theta)]^+=\beta_j(t)-\alpha(t-\theta)= \bigwedge\limits_{i=1}^n [\beta_i(t)-\alpha(t-\theta)]^+\geq 0$.
\end{proof}

\begin{proof}[Proof of Theorem \ref{theorem:non_decreasing_lb_left_over_sc_with_shaped_ac_in_plcs}]
The main idea is as follows.
First, we consider the leftover per step $1 \leq i \leq n$ denoted by $\beta_i \ominus_{\theta} \alpha$  and model its non-decreasing lower bound $\underline{(\beta_i \ominus_{\theta} \alpha)}^\uparrow$  with at least one \curveclass~curve. 
Then, we take the minimum over all these leftover steps according Lemma \ref{lemma:left_over_minima_of_service_curves_equals_minima_of_left_over_service_curves} to obtain the overall leftover curve. \\
The leftover per step $i$ with $\rho_i < +\infty$ and its non-decreasing lower bound is computed follows: \\
Case 1 ($\theta \leq D+\tau_i - \frac{b-L}{R'-r}$): 
Let $y=((D+\tau_i)-\theta)r+b$. 
If $y \leq \sigma_i$, then the leftover for step $i$ is lower bounded by $\delta_D \otimes (((\delta_{\tau_i} \otimes I_{\sigma_i-y} \otimes \gamma_{0, \rho_i-r}) \wedge \delta_0) \wedge ((\delta_{\theta-D}\otimes I_0 \otimes \gamma_{0,+\infty})\wedge \delta_0))$ for $\theta > D$ and $\delta_D \otimes ((\delta_{\tau_i} \otimes I_{\sigma_i-y} \otimes \gamma_{0, \rho_i-r})\wedge \delta_0)$ for $\theta \leq D$. 
If $y > \sigma_i$, then we have to compute the intersection of $\alpha(t-\theta)$ and $\delta_D \otimes ((\delta_{\tau_i}\otimes I_{\sigma_i}\otimes \gamma_{0, \rho_i})\wedge \delta_0)$, so the leftover for step $i$ is lower bounded by $\delta_D \otimes ((\delta_{\tau_i + \frac{y-\sigma_i}{\rho_i -r}} \otimes I_0 \otimes \gamma_{0, \rho_i -r})\wedge \delta_0)$. %\\\\
\begin{figure}[H] 
\begin{centering}
  \includegraphics[width=0.5\textwidth]{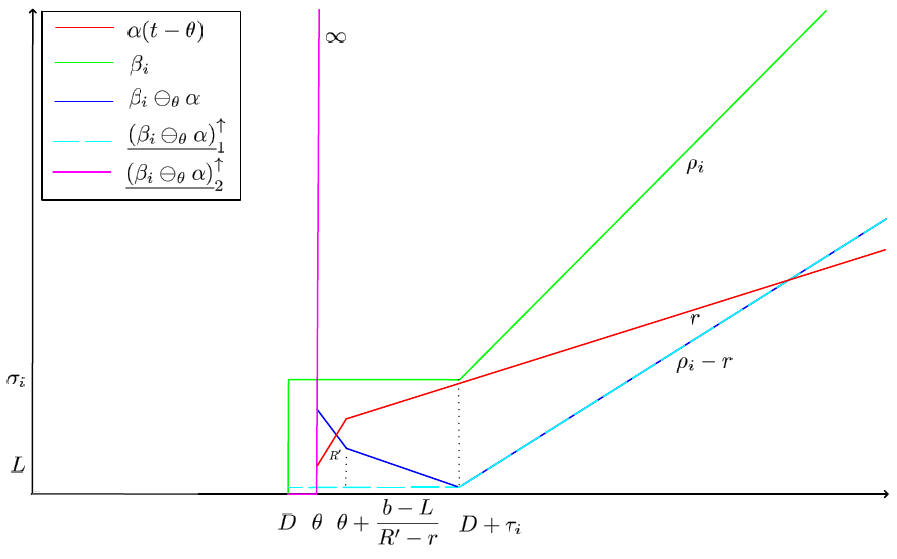}  
  \caption{Example for $\theta$ belonging to Case 1 of the leftover service curve proof.}
\end{centering}
\end{figure}
Case 2 ($\theta > D+\tau_i - \frac{b-L}{R'-r}$): 
First assume $\theta \geq D$. 
Let $y=(\theta + \frac{b-L}{R'-r}-(D+\tau_i))\cdot \rho_i+\sigma_i$. 
If $y \geq \frac{b-L}{R'-r}\cdot r + b$, then the leftover for step $i$ is  lower bounded by $\delta_D \otimes (((\delta_{\theta - D +\frac{b-L}{R'-r}} \otimes I_{y-(\frac{b-L}{R'-r}r+b)}\otimes \gamma_{0,\rho_i-r})\wedge \delta_0)\wedge ((\delta_{\theta-D}\otimes I_0 \otimes \gamma_{0,+\infty})\wedge \delta_0))$. 
If $y < \frac{b-L}{R'-r}\cdot r + b$, then the intersection of $\alpha(t-\theta)$ and $\delta_D \otimes ((\delta_{\tau_i}\otimes I_{\sigma_i}\otimes \gamma_{0, \rho_i})\wedge \delta_0$ needs to be found which will be at time $\theta + \frac{b-L}{R'-r}+\frac{\frac{b-L}{R'-r}r+b-y}{\rho_i-r}$, so the leftover for step $i$ is lower bounded by $\delta_D \otimes ((\delta_{\theta-D+\frac{b-L}{R'-r}+\frac{\frac{b-L}{R'-r}r+b-y}{\rho_i-r}}\otimes I_0 \otimes \gamma_{0, \rho_i-r})\wedge \delta_0)$. \\
Now, we assume $\theta \leq D$. 
Let  $y=(\theta+\frac{b-L}{R'-r}-(D+\tau_i))\cdot \rho_i + \sigma_i$. 
If $y \geq \frac{b-L}{R'-r}r+b$, then the leftover for step $i$ is lower bounded by $\delta_D \otimes ((\delta_{\theta+\frac{b-L}{R'-r}-D}\otimes I_{y-(\frac{b-L}{R'-r}r+b)} \otimes \gamma_{0, \rho_i -r})\wedge \delta_0)$. 
If $y < \frac{b-L}{R'-r}r+b$, then this is the same as in the case $\theta \geq D$. %\\\\
\begin{figure}[H] 
\begin{centering}
  \includegraphics[width=0.5\textwidth]{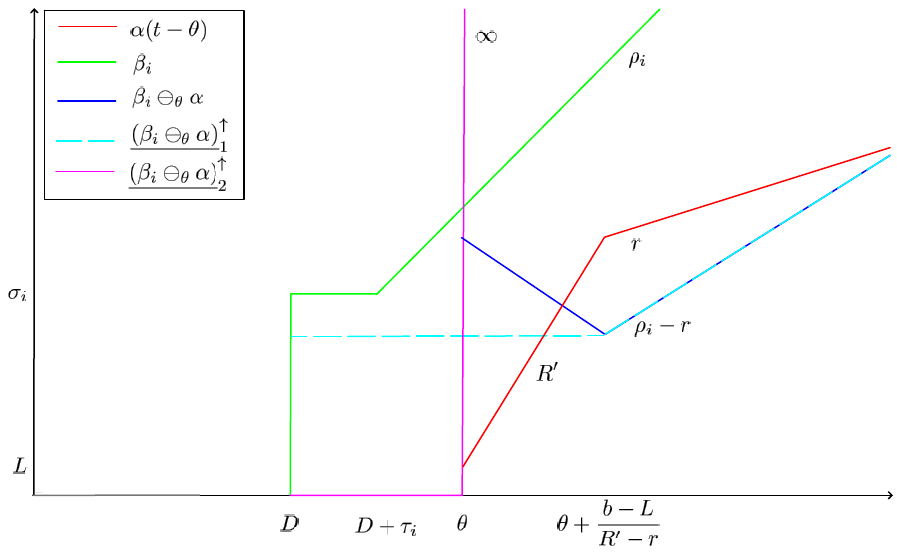}  
  \caption{Example for $\theta$ belonging to Case 2 of the leftover service curve proof.}
\end{centering}
\end{figure}
Lastly, we consider $\rho_i = +\infty$.
If $\theta \leq D$, the leftover for step i is lower bounded by $\delta_{D} \otimes ((\delta_{\tau_i} \otimes I_0 \otimes \gamma_{0,+\infty}) \wedge \delta_0 )$.
If $D \leq \theta \leq D+\tau_i$, the leftover for step i is lower bounded by $\delta_{D} \otimes (((\delta_{\tau_i} \otimes I_0 \otimes \gamma_{0,+\infty}) \wedge \delta_0 ) \wedge ((\delta_{\theta - D} \otimes I_0 \otimes \gamma_{0,+\infty}) \wedge \delta_0 ))$.
If $\theta > D+\tau_i$, the leftover for step i is lower bounded by $\delta_{D} \otimes ((\delta_{\theta - D} \otimes I_0 \otimes \gamma_{0,+\infty}) \wedge \delta_0 ) $.
\\

Hence, the leftover for step $i$ is lower bounded by $\delta_D \otimes (\bigwedge\limits_{j=1}^{\text{size}(i)}((\delta_{\tau_j'} \otimes I_{\sigma_j'} \otimes \gamma_{0, \rho_j'}) \wedge \delta_0))$ with $\text{size}(i) \in \{1,2\}$ depending on the specific case. 
Hence, the lower bound non-decreasing leftover service curve has the form $ \delta_D \otimes (\bigwedge\limits_{i=1}^n \bigwedge\limits_{j=1}^{\text{size}(i)} ((\delta_{\tau_j'} \otimes I_{\sigma_j'} \otimes \gamma_{0, \rho_j'})\wedge \delta_0)) \in$ \curveclass~due to $(f\otimes g)\wedge h = (f\wedge g)\otimes h$ (see \cite{le2001network}) which holds for any functions $f,g,h \in \mathcal{F}$.
\end{proof}

\begin{proof}[Proof of Corollary \ref{cor:symbolic_lb_left_over_sc}]
We first start by arguing that only $\theta \geq D$ is relevant as lower values of $\theta$ yield a smaller leftover service curve.
Indeed, for such values it is clear from the proof of Theorem \ref{theorem:non_decreasing_lb_left_over_sc_with_shaped_ac_in_plcs} that $\theta \uparrow$ yields $y \downarrow$ for Case 1 and $y \uparrow$ for Case 2, so a larger leftover service for the step $i$ at least until $\theta = D$ is reached.
Hence, it is sufficient for the analysis to consider $\theta \geq D$  only.
Next, we can rewrite the non-decreasing lower bound leftover per step $i$ in the proof of Theorem \ref{theorem:non_decreasing_lb_left_over_sc_with_shaped_ac_in_plcs} in a more compact way 
%and note that
%	\begin{itemize}
%		\item for the first (finite) part: In case of $y-\sigma_i \geq 0$, we have to show that $\theta-D \leq \tau_i + \frac{y-\sigma_i}{\rho_i-r}$, but we know $\theta \leq D + \tau_i - \frac{b-L}{R'-r}$. Hence, $\theta -D \leq \tau_i - \frac{b-L}{R'-r} \leq \tau_i + \frac{y-\sigma_i}{\rho_i-r}$.
%		\item for the second (finite) part: In case of $\frac{b-L}{R'-r}r+b-y>0$, note that $\theta -D + \frac{b-L}{R'-r}+\frac{[\frac{b-L}{R'-r}r+b-y]^+}{\rho_i-r} \geq \theta -D$ holds.
%	\end{itemize}
which completes the proof.
\end{proof}

}% End of Appendix

\bibliographystyle{IEEEtran}
\bibliography{literature_journal.bib}

%\bf{If you include a photo:}\vspace{-33pt}
%\begin{IEEEbiography}[{\includegraphics[width=1in,height=1.25in,clip,keepaspectratio]{fig1}}]{Michael Shell}
%Use $\backslash${\tt{begin\{IEEEbiography\}}} and then for the 1st argument use $\backslash${\tt{includegraphics}} to declare and link the author photo.
%Use the author name as the 3rd argument followed by the biography text.
%\end{IEEEbiography}
\vspace{-33pt}
\begin{IEEEbiographynophoto}{Alexander Scheffler}
received the MSc degree in computer science from TU (now RPTU) Kaiserslautern, Germany, in 2020. 
From 2020 to 2023, he worked as a research associate at the Faculty of Computer Science, Ruhr University Bochum, Germany.
In 2022, he was a visiting researcher at Huawei in France.
Since 2024, he has been working as an embedded software engineer in the semiconductor industry.
His primary field of research is the Network Calculus methodology.
\end{IEEEbiographynophoto}

%\section*{What to fix}
%1 $\gamma_{r,b}$ and  $\gamma_{b,r}$ mixed in the paper.Adapt the intro and related work to save work. \\ => DONE
%2 $\gamma_{b,r}$ and  $\gamma_{\sigma,r}$ mixed in the paper. \\
%3 Consistent naming: service curve, leftover, arrival curve, delay and backlog, output bound? \\
%4 Titles of sections (capitel letters)
%5 Figure placement \\
%6 Figure size / aspect ratio \\
%7 Maybe remove some figures or merge similar ones.\\
%8 No PLCS mentioned anymore? \\
%9 $(f\otimes h)\wedge h = (f\wedge g)\otimes h$ several times used, right? Put into a lemma and refer to it. \\ => Not really but similar. Added reference to nc book to both of them.
%10 Check that all theorems need rate inifity treatment. \\
%11 Check that $\delta_0$ not forgotten. \\
%12 Check brackets. \\
%13 piece should be renamed to step. \\
%14 Input shaping -> shaping at endpoint? to not have different wordings \\
%15 foi = flow of interest \\
%16 Connecting sentences all there? \\
%17 grammar check. \\
%18 1x proof read for flow and correctness.

\end{document}